\documentclass[11pt,a4paper]{article}
\usepackage{vmargin}
\setmarginsrb{1.0in}{1.0in}{1.0in}{1.0in}{0mm}{0mm}{5mm}{5mm}

\usepackage[utf8]{inputenc}
\usepackage[T1]{fontenc}
\usepackage{rotating, graphicx}
\usepackage{lscape}

\title{Geometric Network Creation Games\\{\small (full version)}}
\author{Davide Bil{\`{o}}\thanks{Department of Humanities and Social Sciences, University of Sassari, Italy} \and Tobias Friedrich\thanks{Hasso Plattner Institute, University of Potsdam, Germany, \texttt{firstname.lastname@hpi.de}} \and Pascal Lenzner\footnotemark[2] \and Anna Melnichenko\footnotemark[2] }

\usepackage{tabularx} 
\usepackage{amsfonts,amsmath,amssymb,amsthm}
\usepackage{xspace}
\usepackage{hyperref}
\usepackage{multirow}
\usepackage{booktabs} 
\renewcommand{\arraystretch}{1.3} 
\usepackage{graphicx}
\graphicspath{{figures/}}
\usepackage{pgfplots}
\pgfplotsset{compat=1.10}

\usepackage{todonotes}

\usepackage[ruled,vlined,linesnumbered,algo2e]{algorithm2e}
\SetKwFor{Function}{function}{is}{end function}

\SetKw{Input}{input}
\SetKw{Output}{output}
\SetKw{Initialization}{initialization}

\newcommand{\GNCG}{GNCG\xspace} 
\newcommand{\MNCG}{M--GNCG\xspace} 
\newcommand{\RDNCG}{$\mathbf{R}^d$--GNCG\xspace} 
\newcommand{\TNCG}{T--GNCG\xspace} 
\newcommand{\OTNCG}{1-2--GNCG\xspace} 
\newcommand{\OPT}{OPT\xspace} 
\newcommand{\sprofile}{\mathbf{s}} 
\newcommand{\AoE}{AE} 
\newcommand{\PoA}{PoA\xspace}

\newtheorem{theorem}{Theorem}
\newtheorem{lemma}{Lemma}
\newtheorem{corollary}{Corollary}
\newtheorem{conjecture}{Conjecture}

\begin{document}

\maketitle
\begin{abstract}
Network Creation Games are a well-known approach for explaining and analyzing the structure, quality and dynamics of real-world networks like the Internet and other infrastructure networks which evolved via the interaction of selfish agents without a central authority. In these games selfish agents which correspond to nodes in a network strategically buy incident edges to improve their centrality. However, past research on these games has only considered the creation of networks with unit-weight edges. In practice, e.g. when constructing a fiber-optic network, the choice of which nodes to connect and also the induced price for a link crucially depends on the distance between the involved nodes and such settings can be modeled via edge-weighted graphs. 
We incorporate arbitrary edge weights by generalizing the well-known model by Fabrikant et al.~[PODC'03] to edge-weighted host graphs and focus on the geometric setting where the weights are induced by the distances in some metric space.  
In stark contrast to the state-of-the-art for the unit-weight version, where the Price of Anarchy is conjectured to be constant and where resolving this is a major open problem, we prove a tight non-constant bound on the Price of Anarchy for the metric version and a slightly weaker upper bound for the non-metric case. Moreover, we analyze the existence of equilibria, the computational hardness and the game dynamics for several natural metrics. The model we propose can be seen as the game-theoretic analogue of a variant of the classical Network Design Problem. Thus, low-cost equilibria of our game correspond to decentralized and stable approximations of the optimum network design. 
\end{abstract}

\section{Introduction}
Designing efficient networks is a core topic in Computer Science and Operations Research and the study of classical combinatorial optimization problems like the Minimum Spanning Tree Problem~\cite{GP85}, the Steiner Tree Problem~\cite{HRW92} and the Network Design Problem~\cite{JLK78,GJ02,MW84} has a significant impact on these research fields. However, all these problems assume that there is a central authority designing the respective network. In practice, many important infrastructure networks like the physical Internet, the road network and the electricity network are the outcome of a distributed and decentralized design process by many interacting agents. This observation kindled the study of game-theoretic models for network formation by selfish agents. In these models the constructed network is determined by the agents' strategies and the focus is on equilibrium networks, where no agent wants to locally change the network~\cite{Pap01}. The core research question for such models is to quantify the loss of social welfare due to the lack of a central designer and due to the agents' selfishness, i.e., comparing the social cost of the worst possible equilibrium network with the social optimum network~\cite{KP99}. Moreover, also the study of the computational hardness of finding the best possible strategy of an agent and the analysis of the convergence properties of the induced sequential processes are key questions in the field. 

Currently there are two classes of network formation games: variants of the Network Creation Game (NCG)~\cite{Fab03} and Network Design Games (NDG), e.g., \cite{ADKTWR, ADTW08}. In the former games, the selfish agents build incident edges to a subset of other agents to be as central as possible in the constructed connected network. Hence, agents face a trade-off between costs for building and maintaining edges and the service cost for using the network. Here centrality is measured in the created unweighted network which consists of all edges built by the agents and distances are measured as hop-distances. Moreover, every edge has the same fixed price $\alpha > 0$ which is paid by the building agent. Hence, NCGs can be understood as games where a complete unweighted network is given as \emph{host graph} and agents, corresponding to nodes in the host graph, strategically select incident edges in the host graph for the price of $\alpha$ per edge. The constructed network is the sub-network of the host graph which only contains selected edges.  

In contrast, in NDGs a given network with weighted edges serves as the host graph and every agent has a pair of terminal nodes in the host graph she wants to connect. For this, agents select a connecting path in the host graph and pay a cost proportional to the length of the path for its usage. If edges are used by several agents then the cost of the edge is split among these agents.     

Thus in NCGs the distances between all pairs of nodes are important, whereas in NDGs the focus is on simply connecting the terminal pairs. Moreover, the former assume a complete unweighted host graph, whereas the latter assume a weighted not necessarily complete host graph. Hence, NCGs are suitable to model the formation of social networks or the AS-level graph of the Internet, where using the hop-distance is more natural and where agents want to be central, i.e., close to all other agents. But, since NCGs crucially rely on an unweighted host graph, these models cannot be used to investigate the creation of physical communication networks, where edges, e.g., fiber-optic cables, have lengths. 
NDGs are well-equipped to model the creation of physical communication networks between given terminal pairs, e.g., a network connecting many clients to a server or access point, where only connectivity matters. However, NDGs are not suited for studying settings where the agents are interested in communicating with all other agents and where agents are restricted to buy only incident edges.   

To overcome these shortcomings of NCGs and NDGs, we propose and investigate a model which is a generalization of NCGs but which also shares some aspects with NDGs and therefore allows to model the creation of physical communication networks where the goal is to achieve an efficient communication between all pairs of nodes at low cost. That is, we are interested in the decentralized creation of \emph{edge-weighted} networks which minimize the pairwise distances between agents and the total cost of all built edges. This can be seen as the game-theoretic analogue of the well-known Network Design Problem~\cite{JLK78}, where a weighted network and budgets for buying edges and the total routing cost between all pairs are given and the goal is to select a sub-network which respects both budgets. For this, we consider a variant of the NCG, where the given host graph is an arbitrary weighted graph and the prices for buying and using an edge are proportional to its weight. For example, with this we can model the realistic geometric setting where agents have a position in some metric space and the given weighted host graph uses the distance between the positions of the involved agents as edge weights. To the best of our knowledge, this is the first variant of a NCG with weighted edges. 

\subsection{Model and Notation} 
We consider a generalization of the well-known Network Creation Game by Fabrikant et al.~\cite{Fab03}. In our game, called the \emph{Generalized Network Creation Game (\GNCG)}, we consider a given host graph $H = (V,E(H))$, which is a complete undirected weighted graph on $n$ nodes $v_1,\dots,v_n$ with arbitrary non-negative edge weights $w: E(H) \rightarrow \mathbb{R}^+$. Since edges are undirected, we will denote any edge $\{u,v\} \in E(H)$ as $(u,v)$ and we assume $(u,v) = (v,u) = \{u,v\}$. 

Every node of $H$ corresponds to a selfish agent who wants to participate in the network formation. To achieve this, agents strategically decide which subset of incident edges to buy, i.e. a strategy $S_u$ of an agent $u$ is any node subset of $V\setminus\{u\}$ towards which agent $u$ wants to create edges. If $v \in S_u$, then we call agent $u$ the \emph{owner} of the undirected edge $(u,v)$ and $u$ has to pay the full edge price\footnote{If $v \in S_u$ and $u \in S_v$ then both agents have to pay the full edge price. However, in this case one of the agents could improve on her current situation in the network by not buying the edge $(u,v)$, which implies that in any equilibrium or in the social optimum network every edge has exactly one owner.}. 

We assume that the edge price of any edge $(u,v)$ is proportional to its weight $w(u,v)$. In particular, we assume that the edge price for any edge $(u,v)$ is $\alpha \cdot w(u,v)$, where $\alpha > 0$ is a fixed parameter of the game which allows to model different trade-offs between the cost for buying and for using edges. 
 
Let $\sprofile = (S_{v_1},\dots,S_{v_n})$ be any \emph{strategy profile}, which is any vector of strategies of all $n$ agents. The strategy profile $\sprofile$ uniquely determines a subgraph $G(\sprofile) = (V,E(\sprofile))$ of the host graph $H = (V,E(H))$, where $E(\sprofile)=\{(u,v) \mid u \in V, v \in S_u\}$. 

Let $d_G(u,v)$ be the distance between two nodes $u$ and $v$ in the network $G$, which is equal to the total weight of the shortest path between $u$ and $v$, or $+\infty$ if such a path does not exist. To simplify the notation we will use $d_G(u,V):=\sum_{v\in V}{d_G(u,v)}$ and $w(u,S_u):=\sum_{v\in S_u}{w(u,v)}$, where $d_G(u,V)$ is the \emph{distance cost} and  $\alpha\cdot w(u,S_u)$ is the \emph{edge cost} of the agent $u$. By $G-(u,v)$ (or $G+(u,v)$) we denote a network $G$ where edge $(u,v)$ is removed (is added, respectively). 

Given any strategy profile $\sprofile$ and its corresponding network $G(\sprofile)$, then the \emph{cost} of agent $u$ in $G(\mathbf{s})$ is defined as 
$$cost(u,G(\sprofile))=\alpha \cdot w(u,S_u) +d_{G(\sprofile)}(u,V).$$ 
The \emph{social cost} of network $G(\sprofile)$, denoted $cost(G(\sprofile))$ is defined as the sum of the cost of all agents, i.e., $cost(G(\sprofile)) = \sum_{u \in V} cost(u,G(\sprofile))$. 

For any host graph $H$, we say that the \emph{social optimum subgraph} $\OPT$ of $H$ is the network $G(\mathbf{s^*}) = (V,E(\mathbf{s^*}))$ which minimizes $cost(G(\mathbf{s^*}))$ among all possible strategy profiles. Thus, $OPT$ minimizes  
$\alpha \cdot \sum_{(u,v) \in E(\mathbf{s^*})} w(u,v) + \sum_{u\in V}d_{G(\mathbf{s^*})}(u,V)$.

We say that a strategy change from $S_u$ to $S_u'$ is an \emph{improving move} for agent $u$, if $cost(u,G(\sprofile))$ $< cost(u,G(\mathbf{s'}))$, where $\mathbf{s'}$ is identical to $\sprofile$ except for agent $u$'s strategy, which is $S_u'$ instead of $S_u$.  
If there is no improving move for agent $u$ with strategy $S_u$ in $\sprofile$, then we say that $S_u$ is agent $u$'s \emph{best response}. Any strategy change towards a best response strategy is called a \emph{best response move}. A sequence of best response moves which starts and ends with the same strategy vector is called a \textit{best response cycle}. If any sequence of improving moves is finite, then the game has the \emph{finite improvement property (FIP)} which is equivalent to the game being a potential game~\cite{MS96}. The existence of a best response cycle therefore proves that the FIP is violated and thus that the game is not a potential game.  

We say that a network $G(\sprofile)$ is in \emph{pure Nash Equilibrium (NE)}, if no agent in $G(\sprofile)$ has an improving move.
A network $G(\sprofile)$ is in \emph{Greedy Equilibrium (GE)}~\cite{L12} if no agent can improve by buying, swapping or deleting a single edge, where a swap is the combination of deleting an incident edge and buying another one. Moreover, $G(\sprofile)$ is in \emph{Add-only Equilibrium (\AoE)}, if no agent can improve by buying a single incident edge. It directly follows that any network in NE is also in GE and any network in GE is also in \AoE. Additionally, we say that $G(\sprofile)$ is in \emph{$\beta$-approximate NE ($\beta$-NE)} if no agent $u$ can change her strategy to decrease her cost to less than $1/\beta \cdot cost(u,G(\sprofile))$. A \emph{$\beta$-approximate GE ($\beta$-GE)} is defined analogously. 

We measure the impact of selfishness on the quality of the created networks via the \emph{Price of Anarchy}~\cite{KP99}, which for our model is the maximum over the social cost ratios of any NE network and its corresponding social optimum network OPT.

\paragraph{Model Variants}
Besides the \GNCG, where the game is played on a complete host graph $H$ with arbitrary non-negative edge weights, we also consider several interesting special cases (See Fig.~\ref{fig:model_relation} for an overview). 
\begin{figure}[h!]
\center{\includegraphics[width=0.9\textwidth]{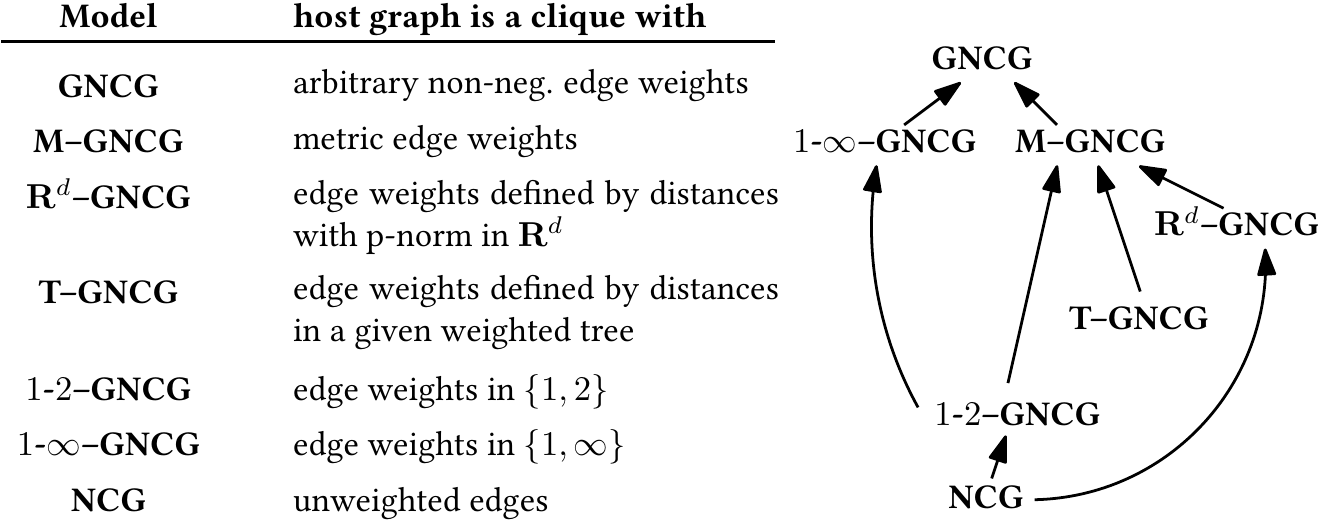}}
\caption{Left: Model overview. Right: Model relations. Arrows point from special case to more general model.}\label{fig:model_relation}
\end{figure}
In the \emph{metric \GNCG}~(\MNCG) the edge weights of $H$ satisfy the triangle inequality. Besides the general metric version, we consider three versions where the edge weights of $H$ are defined by specific metrics. In the simplest case, the \OTNCG, the edge weights of $H$ are restricted to the set $\{1,2\}$. We also consider the variant where the metric edge weights of $H$ are derived from the shortest path distances in a given weighted tree, the \TNCG. Finally, we consider the variant \RDNCG, where the agents are points in $\mathbf{R}^d$ and the edge weights of $H$ correspond to their $p$-norm distances. The original NCG~\cite{Fab03} where $H$ is an unweighted clique, is the most restricted special case of the \MNCG. In the literature a non-metric special case of the \GNCG, where the edge weights are restricted to the set $\{1,\infty\}$ was proposed~\cite{De09}. We call this variant the $1$-$\infty$--GNCG. 

\subsection{Related Work}
There is a huge body of literature both on variants of the Network Creation Game and on Network Design Games and it is impossible to give a full account. Instead, we focus on proposed models which share core features of our model and discuss how they are related to our approach. See also Table~\ref{table:overviewResults}.

The Network Creation Game (NCG) was proposed by Fabrikant et al.~\cite{Fab03} and can be seen as a simplified variant of the connection game by Jackson \& Wolinsky~\cite{JW96}. Due to its simplicity, the NCG became the basis of almost all later modes (including our model). For the NCG researchers have focused mostly on analyzing the PoA and on studying the computational hardness and dynamic properties. A long line of research, e.g.,~\cite{Fab03,De07,Al06,MS10,MMM13,AM17,BL18,AM18}, has established that the PoA of the NCG is constant for almost all $\alpha >0$ and it is widely conjectured that this holds for all $\alpha$. Fabrikant et al.~\cite{Fab03} proved a general upper bound of $\mathcal{O}(\sqrt{\alpha})$, which is still the best known bound which only depends on $\alpha$, and the best known general upper bound as a function of $n$ is $o(n^{\varepsilon})$, for any $\varepsilon >0$, and is due to Demaine et al.~\cite{De07}. It has been shown that computing a best response is NP-hard~\cite{Fab03} and that this holds for many variants of the NCG, e.g.~\cite{MS12,BGLP14,CL15,BGLP16,CLMM16,CLMM17}. However, also restricted variants with efficient best response computation exist, e.g.,~\cite{BG00,ADHL10,L12,BGLP14,Fried17}. Regarding the dynamic properties, it has been shown~\cite{L11,KL13} that many NCG variants do not have the finite improvement property~\cite{MS96}, which states that any sequence of improving strategy changes by agents must be finite. Hence, natural convergence protocols like iterated best response dynamics have no convergence guarantee.

A weighted version of the NCG has been proposed by Albers et al.~\cite{Al06}. In contrast to our model, they consider a version with a specific amount of traffic between each pair of agents but distances are still measured by counting hops. Much closer to our model is the work by Demaine et al.~\cite{De09}, where a NCG on a general unweighted host graph is introduced. This corresponds to the special case of our model where only the edge weights $1$ and $\infty$ are allowed. Edge weight $\infty$ encodes that a particular edge cannot be bought. The authors prove a general upper bound on the PoA of $\mathcal{O}(\sqrt{\alpha})$ and for $\alpha \geq n$ they show that the PoA is in $\Omega\left(\min\left\{\sqrt{\alpha/n},n^2/\alpha\right\}\right)$ and at most $\min\left\{\mathcal{O}\left(\sqrt{n}\right),n^2/\alpha\right\}$ which yields a tight non-constant PoA bound if $\alpha \in \omega(n^{1.5})$ and $\alpha \in o(n^2)$ holds. The highest tight PoA bound as a function of $\alpha$ therefore is $\Theta(\sqrt[5]{\alpha})$ and is achieved for $\alpha = n^{\frac{5}{3}}$. Unfortunately, the proof techniques in \cite{De09} crucially rely on edge weights in $\{1,\infty\}$ and can therefore not be carried over to our model. However, their lower bound construction yields a lower bound of $\Omega(\sqrt[5]{\alpha})$ for the general non-metric case of our model. Also related is the work of Bilò et al.~\cite{Bilo12} who investigated the max-version of the NCG~\cite{De07}, where agents try to minimize their maximum distance, on a general unweighted host graph.

One of the distinctive features of our model is the non-uniform edge price. A few other models with this feature have been proposed, e.g.,~\cite{MMO14, MMO15,CLMM17}, but they all use unit-weight edges. In the model by Cord-Landwehr et al.~\cite{CMadH14} agents can choose different quality levels of an edge for different prices, i.e., the paid price influences the edge length. With this, the model is incomparable to our approach.

Also related are network formation games where not centrality but some other property is the goal of each agent. There are games where agents simply want to be connected to all other agents, e.g.~\cite{BG00,Kli11,Goyal16}. Among them, the work by Eidenbenz et al.~\cite{EKZ06} is closely related to our work. In their wireline strong connectivity game agents are points in the Euclidian plane who strategically buy incident edges to create a connected network. The edge price equals the length of the edge. This is similar to our model in the Euclidian plane with $\alpha=1$ but the focus on connectivity changes the game completely. Another related geometric game was proposed by Moscibroda et al.~\cite{MSW06}. Also there the agents are points in some metric space but agents pay a fixed price for each edge and try to minimize the total stretch towards all other agents. Guly{\'a}s et al.~\cite{Gul15} considered a network formation game in the hyperbolic plane where agents strive for maximum navigability. This is also a geometric model but drastically different from our approach.  

Network Design Games have been proposed in~\cite{ADKTWR, ADTW08}. Their most important feature is that they are potential games~\cite{MS96}, which already shows the contrast to Network Creation Games. Interestingly, Hoefer \& Krysta~\cite{HK05} proposed and analyzed a geometric version.

There are many classical optimization problems related to network design, e.g. see the survey by Magnanti \& Wong~\cite{MW84}. Many of them are NP-complete, e.g. all the problems labeled ``ND'' in~\cite{GJ02}. Our model is closely related to the Network Design Problem~\cite{JLK78} and the Optimum Communication Spanning Tree Problem~(ND7 in \cite{GJ02}). In particular, finding the social optimum network corresponds to a variant of the Network Design Problem, where, instead of having separate budgets for buying edges and for the routing cost, the sum of edge costs and routing costs, i.e., the total distance between all pairs of nodes, is to be minimized. Hence we strongly suspect that computing the social optimum in all versions of our model, with the \OTNCG and the \TNCG as exceptions, is NP-hard.  
\begin{landscape}
\begin{table*}[!ht]%
\caption{Overview of our results and comparison with related work.}
\label{table:overviewResults}
	\begin{center}
		\renewcommand{\arraystretch}{1.5}
		\begin{tabular}{@{}llllll@{}}\toprule
			\textbf{Model} &&\textbf{PoA}  &  \textbf{Complexity} &  \textbf{FIP} & \textbf{Equilibria} \\
			\hline
			\textbf{NCG$^\ast$} 				& &$\mathcal{O}(\sqrt{\alpha})$~\cite{Fab03}, $o(n^\epsilon)$~\cite{De07}	& BR NP-hard~\cite{Fab03}			& no~\cite{KL13}			& NE exists~\cite{Fab03} \\
			\cmidrule{1-6}
			\textbf{1-$\infty$--GNCG$^\ast$}	& & $\Theta\left(\sqrt[5]{\alpha}\right)$~\cite{De09}						& BR NP-hard~(Cor.~\ref{cor_special_cases})		& no~(Cor.~\ref{cor_special_cases})		& ? \\
			\cmidrule{1-6}
			\multirow{5}{*}{\textbf{\OTNCG}} 		& $\alpha <\frac{1}{2}$ 				& = 1 (Thm.~\ref{th:1_2_graph_PoA_is_1})										& \multirow{4}{*}{BR NP-hard (Cor.~\ref{cor_special_cases}),}  	& \multirow{5}{*}{no~(Cor.~\ref{cor_special_cases})}	& NE exists (Thm.~\ref{th:1_2_graph_PoA_is_1}) \\
												& $\frac{1}{2}\leq \alpha < 1$ 						& $=\frac{3}{\alpha+2}$ (Thm.~\ref{thm:LB_PoA_1_2_graph_alpha_1}+\ref{th:1_2_graph_UB_PoA_alpha_between_one_half_and_1}) & \multirow{4}{*}{Dec. NE NP-hard (Thm.~\ref{thm:1_2_graph_des_probl_hard})} && \multirow{2}{*}{NE exists (Thm.~\ref{thm_1-2_graph_alpha_1_NE})}\\
												& $\alpha = 1$ 			& = $\frac{3}{2}$ (Thm.~\ref{thm:LB_PoA_1_2_graph_alpha_1}+ \ref{thm:metric_PoA})	&&& \\
												& $1 < \alpha < 3$ 						& \multirow{2}{*}{$\mathcal{O}(\sqrt{\alpha})$ (Thm.~\ref{thm:1_2_graph_UB_PoA_sqrt_alpha})} &&& $3(\alpha+1)$-NE (Cor.~\ref{cor:6_approx_NE}) \\
												& $\alpha \geq 3$		&&&& NE exists (Thm.~\ref{thm:1_2_graph_star_is_NE})\\
			\cmidrule{1-6}
			\textbf{\TNCG} & &$=\frac{\alpha+2}{2}$ (Thm.~\ref{thm:LB_PoA_tree_metric}+\ref{thm:metric_PoA}) & BR NP-hard (Thm.~\ref{thm_tree_hardness}) & no (Thm.~\ref{thm_tree_no_FIP}) & NE exists (Cor.~\ref{thm:T_is_OPT}) \\
			\cmidrule{1-6}
			\multirow{4}{*}{\textbf{\RDNCG}} &\multirow{2}{*}{$p$-norm, $p\geq 2$}  &$\geq\frac{3\alpha^3 + 24\alpha^2 + 40\alpha + 24}{\alpha^3 + 10\alpha^2 + 32\alpha + 24}$ (Thm.~\ref{thm:LB_PoA_Rd_for_arb_p_norm}), &\multirow{4}{*}{BR NP-hard~(Thm.~\ref{th:points_in_he_plane_BR_hardness})} & \multirow{2}{*}{?}  & \multirow{4}{*}{$3(\alpha+1)$-NE (Cor.~\ref{cor:6_approx_NE})}\\
			 & & $\leq \frac{\alpha+2}{2}$ (Thm.~\ref{thm:metric_PoA})& & & \\
			 & \multirow{2}{*}{1-norm}	&$1+\frac{\alpha}{2+\alpha/(2d-1)}$(Thm.~\ref{thm:LB_PoA_Rd_for_1_norm}),	& &\multirow{2}{*}{no (Thm.~\ref{th:points_in_the_plane_best_resp_cycle})} \\
			 &							&$\leq \frac{\alpha+2}{2}$ (Thm.~\ref{thm:metric_PoA}) & & &\\
			\cmidrule{1-6}
			\multirow{2}{*}{\textbf{\MNCG}} & &\multirow{2}{*}{$=\frac{\alpha+2}{2}$ (Thm.~\ref{thm:LB_PoA_tree_metric}+\ref{thm:metric_PoA})}  & BR NP-hard~(Cor.~\ref{cor_special_cases})  & \multirow{2}{*}{no~(Cor.~\ref{cor_special_cases})} & \multirow{2}{*}{$3(\alpha+1)$-NE (Cor.~\ref{cor:6_approx_NE})}\\
			&&&Dec. NE NP-hard (Thm.~\ref{thm:1_2_graph_des_probl_hard})\\
			\cmidrule{1-6}
			\multirow{2}{*}{\textbf{\GNCG}} & &$\geq\frac{\alpha+2}{2}$ (Thm.~\ref{thm:LB_PoA_tree_metric}) & BR NP-hard~(Cor.~\ref{cor_special_cases}) & \multirow{2}{*}{no~(Cor.~\ref{cor_special_cases})} & \multirow{2}{*}{?}\\ 
			&& $\leq\left(\frac{\alpha+2}{2}\right)^2$ (Thm.~\ref{thm_general_PoA})&Dec. NE NP-hard (Thm.~\ref{thm:1_2_graph_des_probl_hard})\\
			\bottomrule
		\end{tabular}
	\end{center}
\end{table*}
\end{landscape}
\newpage

\subsection{Our Contribution}
In this paper we investigate the classical Network Creation Game on edge-weighted host graphs. This variant allows modeling the decentralized creation of networks, like fiber-optic communication networks or many variants of overlay networks, by selfish agents, e.g., ISPs. In such settings, the nodes in a network have a physical location and the edge weights and also the cost for creating and maintaining them depend on these locations. In particular, we focus on specific natural metrics, e.g., graph and tree metrics as well as the geometric setting where the agents correspond to points in~$\mathbf{R}^d$. 

We show that computing a best response strategy is NP-hard for all variants of our model and we prove for the \OTNCG that deciding if a given strategy profile is in NE is NP-hard as well. The latter is the first result of this type in the realm of NCGs. Moreover, we prove that all our models do not have the finite improvement property. On the positive side, we give an efficient algorithm for computing a social optimum network for the \OTNCG and we show how to trivially obtain the social optimum in the \TNCG. 

Our main focus is a rigorous study of the quality of the induced equilibrium networks of our models. For this we show that NE exist in the \OTNCG and the \TNCG and that the more general \MNCG always admits a $3(\alpha+1)$-approximate NE. The main contribution of our paper is a collection of bounds on the Price of Anarchy, i.e., we bound the loss in social welfare due to selfishness and to the lack of central coordination. We prove a tight PoA bound of $(\alpha+2)/2$ for the \MNCG and the \TNCG. This bound is remarkable, since it is non-constant and much higher than the previously known upper bounds for the NCG or the inherently non-metric $1$-$\infty$--GNCG. This shows that allowing weighted edges completely changes the picture.   Moreover, in contrast, settling the PoA for the original NCG, which is a special case of all our models, is a major open problem in the field. For the model variant which is closest to the NCG, the \OTNCG, we prove a tight constant bound on the PoA for $\alpha \leq 1$ and show that the PoA is in $\mathcal{O}(\sqrt{\alpha})$ for $\alpha > 1$. Hence, this model behaves very similar to the NCG. For the variant with points in $\mathbf{R}^d$, the \RDNCG, with the 1-norm we show how to embed our lower bound construction from the \TNCG. This yields a tight PoA bound if $d$ tends to infinity. Additionally, for any $p$-norm with $p\geq 2$ we give a lower bound construction which yields PoA of at least $3$ for high $alpha$ and which generally shows that the PoA is larger than $1$. Finally, for the most general case, the \GNCG, we show that the PoA is between $(\alpha+2)/2$ and $((\alpha+2)/2)^2$. 

See Table~\ref{table:overviewResults} for an overview over the majority of our results and the most relevant results for the earlier models which are marked with the star symbol. All results on the Price of Anarchy with an equality sign are tight bounds.

\section{Preliminaries}\label{sec_prelims}
We start by clarifying the relation of the models we investigate. 
Fig.~\ref{fig:model_relation} shows which models are special cases of other models. These relationships and the facts that computing a best response strategy is NP-hard for the NCG~\cite{Fab03} and that the NCG does not have the FIP~\cite{KL13} directly yields the following corollary.
\begin{corollary}\label{cor_special_cases}
Computing a best response strategy is NP-hard for the \OTNCG, $1$-$\infty$--GNCG, the \MNCG and the \GNCG. Additionally, these models do not have the FIP.
\end{corollary}
Let $k \geq 1$. We say that a subgraph $G$ of $H$ is a {\em $k$-spanner} if $d_{G}(u,v) \leq k d_H(u,v)$ for every pair of vertices $u, v \in V$. Next, we show a useful property, which holds for any host graph.
\begin{lemma}\label{lem:NE_spanner}
For any host graph $H$ any \AoE~ is a $(\alpha+1)$-spanner.
\end{lemma}
\begin{proof}
First, we consider edges $(u,v)$ with $w(u,v) = d_H(u,v)$, that is, a shortest path between $u$ and $v$ in the host graph $H$ uses the direct edge. We claim for such pairs $u$ and $v$ that in any NE network $G$ we have $d_G(u,v) \leq (\alpha+1)d_H(u,v) = (\alpha+1)w(u,v).$ To see this, assume towards a contradiction that $d_G(u,v) > (\alpha+1)d_H(u,v)$, which implies that $(u,v) \notin E(G)$. Now consider what happens if agent $u$ buys the edge $(u,v)$: Agent $u$ additionally has to pay $\alpha \cdot w(u,v)$ for creating the edge and then her distance to $v$ is guaranteed to be $w(u,v)$. Thus her total cost for buying the edge $(u,v)$ and reaching node $v$ is $(\alpha+1)w(u,v)$. Since $d_G(u,v) > (\alpha+1)w(u,v)$, buying the edge $(u,v)$ is an improving move for agent~$u$. 

Now we consider two arbitrary agents $u$ and $v$ in $G$ and let $P_{uv} = x_1,x_2,\dots,x_k$ with $u=x_1$ and $x_k=v$ be a shortest path between $u$ and $v$ in the host graph $H$. It follows that $d_H(u,v) = w(x_1,x_2) + w(x_2,x_3) + \dots + w(x_{k-1},x_k)$. Since $P_{uv}$ is a shortest path in $H$ and since any subpath of a shortest path must be a shortest path itself it follows that for all pairs $x_i$ and $x_{i+1}$, with $1\leq i\leq k-1$, the equality $w(x_i,x_{i+1}) = d_H(x_i,x_{i+1})$ holds. Thus, in any NE $G$ on the host graph $H$ we have that $d_G(x_i,x_{i+1}) \leq (\alpha+1)w(x_i,x_{i+1})$ holds for all $1\leq i \leq k-1$. Thus, the distance between $u$ and $v$ in any NE $G$ is $d_G(u,v) \leq (\alpha+1)w(x_1,x_2) + (\alpha+1)w(x_2,x_3) + \dots + (\alpha+1)w(x_{k-1}x_k) = (\alpha+1)d_H(u,v).$
\end{proof}
\noindent With a similar technique we get an analogous statement for the social optimum network $OPT$.

\begin{lemma}\label{lemma_OPT_spanner}
The social optimum network is a $\left(\frac{\alpha}{2}+1\right)$-spanner for any connected host graph $H$. 
\end{lemma}
\begin{proof}
The proof is analogous to the proof of Lemma~\ref{lem:NE_spanner}. 
Let $OPT(H)$ be the subgraph of $H$ which minimizes the social cost. We start by considering edges $(u,v)$ in $OPT(H)$ where $w(u,v) = d_H(u,v)$, that is, a shortest path between $u$ and $v$ in the host graph $H$ uses the direct edge. We claim for such pairs $u$ and $v$ that in $OPT(H)$ we have $$d_{OPT(H)}(u,v) \leq \left(\frac{\alpha}{2}+1\right)d_H(u,v) = \left(\frac{\alpha}{2}+1\right)w(u,v).$$ 
To see this, assume towards a contradiction that $d_{OPT(H)}(u,v) > \left(\frac{\alpha}{2}+1\right)\cdot w(u,v)$, which implies that $(u,v) \notin E(OPT(H))$. Now consider what happens if the edge $(u,v)$ is added to $OPT(H)$: The social cost increases by $\alpha \cdot w(u,v)$ for creating the additional edge. Moreover, the creation of the edge $(u,v)$ ensures that the distance between $u$ and $v$ is $w(u,v)$. Thus, the distance from $u$ to $v$ is decreased by more than $\left(\frac{\alpha}{2}+1\right)w(u,v) - w(u,v) = \left(\frac{\alpha}{2}\right)w(u,v)$. The same holds true for the distance from $v$ to $u$. Thus, the total distance decrease induced by the addition of the edge $uv$ to $OPT(H)$ is more than $2 \left(\frac{\alpha}{2}\right)w(u,v) = \alpha \cdot w(u,v)$. Since the total distance decrease is strictly larger than the edge cost of the edge $(u,v)$, this implies that the network $OPT(H)$ augmented by the edge $(u,v)$ has strictly less social cost than $OPT(H)$. This contradicts the assumption that $OPT(H)$ minimizes the social cost. 

Now we consider two arbitrary agents $u$ and $v$ in $OPT(H)$ and let $P_{uv} = x_1,x_2,\dots,x_k$ with $u=x_1$ and $x_k=v$ be a shortest path between $u$ and $v$ in the host graph $H$. It follows that $d_H(u,v) = w(x_1,x_2) + w(x_2,x_3) + \dots + w(x_{k-1},x_k)$. Since $P_{uv}$ is a shortest path in $H$ and since any subpath of a shortest path must be a shortest path itself it follows that for all pairs $x_i$ and $x_{i+1}$, with $1\leq i\leq k-1$, the equality $w(x_i,x_{i+1}) = d_H(x_i,x_{i+1})$ holds. Thus, in $OPT(H)$ we have that $d_{OPT(H)}(x_i,x_{i+1}) \leq \left(\frac{\alpha}{2}+1\right)w(x_i,x_{i+1})$ holds for all $1\leq i \leq k-1$. Thus, the distance between $u$ and $v$ in $OPT(H)$ is 
\begin{align*}d_{OPT(H)}(u,v) &\leq \left(\frac{\alpha}{2}+1\right)w(x_1,x_2) + \left(\frac{\alpha}{2}+1\right)w(x_2,x_3) + \dots \\ &+ \left(\frac{\alpha}{2}+1\right)w(x_{k-1}x_k)
= \left(\frac{\alpha}{2}+1\right)d_H(u,v). 
\end{align*}
\end{proof}

\section{Host Graphs with Metric Weights}\label{sec_metric}
In this section we investigate the NCG on complete host graphs with edge weights which satisfy the triangle inequality. After giving some general results, we focus on specific natural metrics.
\begin{theorem}\label{thm:metric_PoA}
The PoA in the \MNCG is at most $\frac{\alpha+2}{2}$.
\end{theorem}
\begin{proof}
Let $G$ be a NE and let $u$ and $v$ be two distinct vertices. Let $x$ and $x^*$ be two Boolean variables such that $x=1$ if and only if $(u,v)$ is an edge of $G$ and $x^*=1$ if and only if $(u,v)$ is an edge of the social optimum~$OPT$. We prove the claim by showing that 
\[
\sigma:=\frac{\alpha \cdot w(u,v) \cdot x+2d_G(u,v)}{\alpha \cdot w(u,v)\cdot x^*+2d_{OPT}(u,v)}\leq \frac{\alpha+2}{2}.
\]
Essentially $\sigma$ is the ratio of the social cost contribution of every pair of nodes in the NE and in OPT. If the ratio for every pair of nodes is bounded by $(\alpha+2)/2$ then this also holds for their sum.

Now we prove the claim. If $x=1$ then $d_G(u,v) = w(u,v)$ and hence $\sigma\leq (\alpha+2)\cdot w(u,v)/(2d_{OPT}(u,v))\leq (\alpha+2)\cdot w(u,v)/(2\,w(u,v)) =(\alpha+2)/2$. If $x=0$ and $x^*=1$ then $\sigma\leq 2(\alpha+1)/(\alpha+2) \leq (\alpha+2)/2$ since, by Lemma~\ref{lem:NE_spanner}, $d_G(u,v)\leq (\alpha+1)w(u,v)$.

It remains to prove $\sigma \leq (\alpha+2)/2$ when $x=0$ and $x^* = 0$. This means that there is a vertex $z$ with $z\neq u$ and $z\neq v$ along a fixed shortest path in $OPT$ between $u$ and $v$. As $G$ is a NE, neither $u$ nor $v$ has an incentive to buy the edge towards $z$. If $u$ bought the edge $(u,z)$ at the price of $\alpha\cdot w(u,z)$, its distances towards $z$ would be at most $w(u,z)$ and, by the triangle inequality, its distance towards $v$ would be at most $w(u,z) + d_G(z,v)$. Since this is not an improvement, we have
$(\alpha+2)d_{OPT}(u,z)+d_G(z,v) = (\alpha+2)w(u,z) + d_G(z,v) \geq \alpha\cdot w(u,z)+d_{G+(u,z)}(u,z)+d_{G+(u,z)}(u,v)\geq d_G(u,z)+d_G(u,v)$ and hence \begin{equation}(\alpha+2)d_{OPT}(u,z)+d_G(z,v) \geq d_G(u,z)+d_G(u,v).\label{eq_u}\end{equation}
Analogously for agent $v$ we get 
$(\alpha+2)d_{OPT}(v,z)+d_G(z,u) \geq (\alpha+2)w(v,z)+ d_G(z,u)\geq \alpha\cdot w(v,z) + d_{G+(v,z)}(v,z)+d_{G+(v,z)}(v,u) \geq d_G(v,z)+d_G(v,u)$ which yields 
\begin{equation}
 (\alpha+2)d_{OPT}(v,z)+d_G(z,u) \geq d_G(v,z)+d_G(v,u). \label{eq_v}
\end{equation}
By summing up the inequalities (\ref{eq_u}) and (\ref{eq_v}), we obtain
\begin{align*}(\alpha+2)d_{OPT}(u,v) &= (\alpha+2)d_{OPT}(u,z) + (\alpha+2)d_{OPT}(v,z)\geq 2d_G(u,v).\end{align*}
 Therefore, also the last case yields $\sigma \leq (\alpha+2)/2$. 
\end{proof}

\paragraph{Existence} It is an interesting open question if NE always exist for the \MNCG. Here we prove a weaker result which essentially states that for low $\alpha$ there always is an outcome of the game where no agent can improve by a high multiplicative factor. This yields that there always is a network which is approximately stable.
\begin{theorem}\label{th:2approx_NE}
Any \AoE~ network in the \MNCG is in $(\alpha+1)$-GE.     
\end{theorem}
\begin{proof}
Consider a network $G=(V,E)$ which is in \AoE.
By the definition of a $(\alpha+1)$-GE we need to evaluate the maximal improvement of the cost function which can be made by a deletion or swap of any edge in $G$.

First, we consider a deletion. Compare the cost function value of some agent $u\in V$ before and after an improving deletion of one of her edges $e=(u,v)\in E(G)$:
\begin{align*}
\frac{cost(u,G)}{cost(u,G')}=\frac{\alpha\cdot w(u,S_u) + d_G(u,V)}{\alpha\cdot w(u,S_u') + d_{G'}(u,V)},
\end{align*}    
where $G'$ is the network obtained from $G$ by applying new strategy $S'_u=S_u\setminus \{v\}$, i.e. $E(G')=E(G)\setminus \{(u,v)\}$. In the worst case, the deletion of the edge $(u,v)$ does not change distance between the nodes $u$ and $v$, i.e., $w(u,v) = d_{G'}(u,v)$. This yields 
\begin{align}\frac{cost(u,G)}{cost(u,G')} &= \frac{\alpha\cdot w(u,S_u) + d_G(u,V)}{\alpha\cdot w(u,S_u') + d_G(u,V)}\notag\\
&= \frac{\alpha\cdot w(u,v) + \alpha \cdot w(u,S_u\setminus\{v\}) + d_G(u,V)}{\alpha \cdot w(u,S_u\setminus\{v\}) + d_G(u,V)}\notag\\
&= 1 + \frac{\alpha \cdot w(u,v)}{\alpha \cdot w(u,S_u\setminus\{v\})+d_G(u,V)}\notag\\
&\leq 1 + \frac{\alpha\cdot w(u,v)}{d_G(u,V)} \leq 1 + \frac{\alpha\cdot w(u,v)}{w(u,v)} = 1+\alpha. \label{ineq_deletion}
\end{align}
Now we consider an  improvement which can be made by one swap. Let agent $u\in V(G)$ can improve her cost by swap an edge $(u,v)$ to $(u,w)$, and let $G_{swap}$ be the new graph. Compare the cost function after the swap with the cost value after the sequential addition of the edge $(u,v)$ and the deletion of the edge $(u,w)$. Let $G_{add}$ and $G_{del}$ be the corresponding networks. Thus, $E(G_{swap}) = E(G_{del}) = E(G_{add})\setminus \{(u,w)\} =\left(E(G)\cup \{(u,v)\}\right)\setminus \{(u,w)\}$. Then, by the inequality (\ref{ineq_deletion}) and because $G$ is in \AoE, we have:
\begin{align}
cost_{G_{swap}}(u) = cost_{G_{del}}(u) \geq \frac{1}{\alpha+1}cost_{G_{add}}(u) \geq \frac{1}{\alpha+1}cost_{G}(u). \label{inq_swap}
\end{align}

\noindent Finally, by (\ref{ineq_deletion}) and (\ref{inq_swap}), we get that $G$ is in $(\alpha+1)$-GE.

\end{proof}

\noindent Now, we adapt the technique from~\cite{L12} to relate GE and $\beta$-NE.
\begin{theorem}\label{th:GE_is_3_approx_NE}
In the \MNCG every network in GE is in 3-NE.
\end{theorem}
\begin{proof}
We prove the claim by a "locality gap preserving" reduction to the {\em Uncapacitated Metric Facility Location} problem (UMFL). Roughly speaking, in UMFL we are given a set of facilities, each of which has a non-negative opening cost, a set of clients, and a distance between each client and each facility (the distances satisfy the quadrangle inequality). The task in UMFL it to open a set of facilities and assign each client to the closest opened facility in such a way that the overall cost -- i.e., the overall cost of the opened facilities plus the overall sum of client-to-assigned-facility distances -- is minimized.  Since it was shown in~\cite{arya2004local} that the locality gap of UMFL is 3, that means  that any UMFL solution that cannot be improved by a single move, i.e., by opening, closing or swapping one facility, is a 3-approximation of the optimal solution.

Consider a graph $G=(V,E)$. Let $u\in V$ be an agent in $(G,\alpha)$ and let $Z\subset V$ be the set of vertices which own an edge to $u$. Consider the subgraph $G'=(V,E')$ of $G$ which does not contain edges owned by the agent $u$. Denote $S(u)$ be the set of $u$'s pure strategies in $(G',\alpha)$. We construct an instance $I(G')$ for UMFL from the graph $G'$ as follows: let $F=C=V\setminus\{u\}$, where $F$ is the set of facilities, $C$ is the set of clients; we define for all facilities $f\in Z\cap F$ the opening cost $c(f)$ to be 0, and $c(f)=\alpha\cdot w(f,u)$ for all other facilities. We define distances for all $i\in F, j\in C$ to be $d_{ij}=d_{G'}(i,j)+w(i,u)$. If $G'$ is disconnected, then $d_{ij}=\infty$. 

Now we construct a map $\pi:S(u)\rightarrow S_{UMFL}$, where $S_{UMFL}$ is the set of solutions of the UMFL for the instance $I(G')$, as follows: for any $S\in S(u)$, define $\pi(S)=S\cup Z$ and for any $F_S\in S_{UMFL}$, $\pi^{-1}(F_S)=F_S\setminus Z$. Since for any solution $F_S$ of the UMFL  $Z\subseteq F_S$, the strategy $S'=\pi^{-1}(F_S)$ exists, and for any two strategies $S_1\neq S_2$, $\pi(S_1)\neq \pi(S_2)$. Therefore, the map $\pi$ is a bijection. To prove the statement of the theorem we need to show that if agent $u$ cannot improve her strategy by adding, deleting or swapping one edge, then the corresponding solution $F_S=\pi(S)$ for UMFL cannot be improved by opening, closing or swapping one facility. 

First, we show that the cost of agent $u$ is equal to the cost of the corresponding UMFL solution $F_S$. Indeed, 
\begin{align*}
cost(u,G(S)) &=\alpha\!\cdot\!w(u,S) + \sum\limits_{v\in V\setminus \{u\}}{\left(\min\limits_{x\in S\cup Z}{(d_{G'}(x,v)}+w(u,x))\right)}\\
&= \alpha\!\cdot\!w(u,S\setminus Z) +0\cdot w(u,Z) + \sum\limits_{v\in V\setminus \{u\}}{\left(\min\limits_{x\in S\cup Z}{d_{xv}}\right)}\\
&= \alpha\cdot\sum\limits_{f\in F_S\setminus Z}{c(f)} +\sum\limits_{f\in Z}{c(f)} + \sum\limits_{v\in C}{\left(\min\limits_{x\in F_S}{d_{xv}}\right)}\\
 &= cost(F_S).
\end{align*}

Next we show that $F_s=\pi(S)$ cannot be decreased by opening, closing or swapping one facility. By the sake of contradiction, assume that the solution $F_S$ can be improved by a single step. Denote $F'_S$ be an improved solution. Note that no facility $z\in Z$ is included in an opening, closing or swapping step. Indeed, by construction, $Z\subseteq F_S$ and if there is a facility $z\in F_S\setminus F'_S$, then there is at least one client $c\in C$ such that $d_{cz}\leq \min\limits_{f\in F_S}{d_{cf}}$, thus, closing the facility $z$ does not decrease $cost(F_S)$ and, therefore, $z\in F'_S$. Thus, and because $\pi$ is a bijection, we have that there is a strategy $S'=\pi^{-1}(F'_S)$ such that $S'\neq S$. Therefore, we have $cost(u,G(S'))=cost(F'_S)<cost(F_S)=cost(u,G(S))$. Hence, there is the better strategy $S'$ for the agent $u$, which contradicts with the assumption that there is no one step improvement of the strategy~$S$.

Finally, applying the result by Arya et al.\cite{arya2004local}, we get  $cost(u,G(S))\leq 3\,cost(u, G(S^*))$ where $S^*$ is an optimal strategy in $(G',\alpha)$.
\end{proof}

\noindent By Theorem~\ref{th:2approx_NE} and Theorem~\ref{th:GE_is_3_approx_NE}, we get the following:

\begin{corollary}\label{cor:6_approx_NE}
Every network which is in AE in the \MNCG is in $3(\alpha+1)$-NE.
\end{corollary}

\subsection{1-2-Graphs}\label{subsec_12}
Here we consider the \MNCG for the special case where for every pair of nodes $u$ and $v$ we have either $w(u,v) = 1$ or $w(u,v)=2$. We call an edge of weight 1 or 2 a \textit{1-edge} or \textit{2-edge}, respectively. We call such graphs \emph{1-2-graphs}.  

Studying 1-2-graphs is especially interesting since this class of host graphs is the simplest generalization of the unweighted host graphs from the NCG and the edge weights are guaranteed to satisfy the triangle inequality. 1-2-graphs are commonly used as the simplest non-trivial metric special case, e.g. when studying the TSP~\cite{Karp72, BK06, AMP18}, and hence they are a natural starting point. 

\noindent We start with a simple statement about 1-edges. We show that for $\alpha < 1$ any NE must contain all the 1-edges from the host graph. If $\alpha = 1$, then there always exists a NE which contains all 1-edges.
\begin{lemma}\label{lem:1_edge_property}
For $\alpha = 1$ in any NE network in the \OTNCG buying any additional 1-edge is cost neutral for the buyer. For $\alpha < 1$ buying any 1-edge is an improving move for the buyer.
\end{lemma}
\begin{proof}
Consider a graph $G$ which is in NE in the \OTNCG. Assume there is an edge $(u,v)$ of weight 1 which is not in $G$. Thus, $d_G(u,v)\geq 2$. Then buying the edge by one of its endpoint costs $\alpha$ while the distance cost decreases by at least $2-1$. Hence, if $\alpha < 1$, the decrease of the distance cost exceeds the increase in the edge cost, which means that this is an improving move for the buying agent.  If $\alpha = 1$, the cost for the buying agent does not change.
\end{proof}

\paragraph{Hardness} Here we discuss the hardness of deciding if a given strategy profile is in NE for the \OTNCG. Note that the NP-hardness of computing a best response strategy for some agent, which is guaranteed by Corollary~\ref{cor_special_cases}, does not directly imply the NP-hardness of the NE decision problem. 

First, we take a detour via the Vertex Cover problem. A {\em vertex cover} of an undirected graph $G$ is a subset $C$ of vertices of $G$ such that, for every edge $(u,v)$ of $G$, $u \in C$ or $v \in C$. It is well-known that computing a minimum vertex cover of a subcubic graph is NP-hard. We start with a result which may be folklore.
\begin{lemma}\label{lm:np_completeness_vertex_cover_variant}
Unless P=NP, there is no polynomial time oracle that, given a graph $G$ and a vertex cover of $G$ of size $k$, decides whether $G$ admits a vertex cover of size at most $k-1$.
\end{lemma}
\begin{proof}
We prove the claim by showing that the existence of such an oracle would imply the existence of a polynomial time algorithm for computing a vertex cover of $G$ of size at most $k-1$, assuming it exists. Therefore, by reiterating the algorithm at most $k$ times, we might be able to compute a minimum vertex cover of $G$ in polynomial time, thus proving that P=NP. Let $C$ be a vertex cover of $G$ of size $k$. 

The algorithm works as follows. First of all, we query the oracle to understand whether $G$ admits a vertex cover of size strictly better than $k$. In case of a ``no'' answer, we know that $C$ is an optimal vertex cover and therefore $G$ does not admit a vertex cover of size $k-1$. So, we assume that the oracle answers ``yes''. This implies that there is a vertex cover of size $k-1$.
In the following we show how to compute a vertex cover of size (at most) $k-1$ in polynomial time.

Let $G-v$ be the graph obtained from $G$ without the vertex $v$ (and all the edges incident to $v$). For every vertex $v$ of $C$, we query the oracle using the graph $G-v$ and the cover $C-v$ (so we want to know whether $G-v$ has a vertex cover of size $k-2$).
If all the $k$ answers returned by the oracle are ``no'', then $V(G)\setminus C$ is a vertex cover of size strictly smaller than $k$. Indeed, the answer ``no'' for $v$ means that there is no vertex cover of size $k-1$ that contains $v$. However, since a vertex cover of size $k-1$ exists, such a vertex cover has to contain the entire neighborhood of $v$ (otherwise some edges incident in $v$ would remain uncovered).

To complete the proof, we assume that the oracle has answered ``yes'' for at least one vertex of $C$, say $v$. This means that there is a vertex cover of size (at most) $k-1$ that contains $v$. We build such a vertex cover by adding $v$ to the vertex cover of size (at most) $k-2$ that is computed recursively on $G-v$ and $C-v$. Clearly, the running time of this algorithm is polynomial in  the number of vertices of the graph.
\end{proof}

\begin{theorem}\label{thm:1_2_graph_des_probl_hard}
Unless P=NP, there is no polynomial time algorithm that decides whether a strategy profile is in NE for the \OTNCG.
\end{theorem}
\begin{proof}
The reduction is from the Vertex Cover problem and $\alpha = 1$. More precisely, we define both a 1-2-graph and a strategy profile such that every agent but one is playing her best response and computing a best response of the remaining agent is equivalent to computing a minimum vertex cover.

We define the graph $G=(V,E)$ such that there is one vertex node $a_i\in V$ for each vertex $v_i$ of the Vertex Cover instance, and two edge nodes $p_j$ and $p_j'$ in $V$ for each edge $e_j$ of the Vertex Cover instance. Finally, there is a new node $u$, that is neither a vertex node nor an edge node. There is an edge of weight 1 between vertex node $a_i$ and each edge node $p_j, p_j'$ if and only if $v_i$ is an endvertex of $e_j$. Furthermore, there is an edge of weight 1 between every pair of vertex nodes. All the other edges have weight 2. See Fig.~\ref{fig:NP_1_2_graph_metric} for the construction. 
\begin{figure}[h!]
\center{\includegraphics[width=0.4\textwidth]{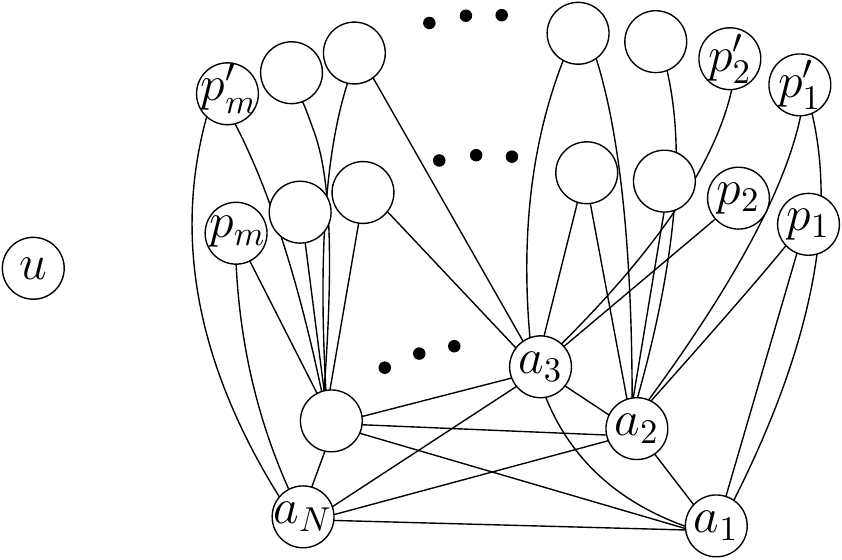}}
\caption{Illustration of the construction used in the reduction. All depicted edges have weight equal to 1; the missing edges are all of weight equal to 2.}\label{fig:NP_1_2_graph_metric}
\end{figure} 

\noindent Consider the strategy profile in which each edge of weight 1 is bought by any of the two agents that are incident to the edge, while $u$ is buying all the edges towards vertex nodes that correspond to a vertex cover of size $k$ w.r.t. the Vertex Cover instance computed using any polynomial time algorithm. Note that by Lemma~\ref{lem:1_edge_property} and since $\alpha =1$, buying a 1-edge is neutral for the incident agents.   

First of all, we observe that the eccentricity of each node is at most 3. 
Therefore, every agent other than $u$ is actually playing a best response. We claim that for any improving move of $u$, there exists another improving move in which agent $u$ buys only the edges towards the vertex nodes that correspond to a vertex cover of size at most $k-1$ w.r.t. the Vertex Cover instance. The claim then would follow from Lemma~\ref{lm:np_completeness_vertex_cover_variant}.

Consider any improving move $S_u$ for $u$. We prove the claim by first showing the existence of an alternative improving move consisting only of edges towards vertex nodes. Indeed, if $u$ bought an edge towards an edge node in $S_u$, w.l.o.g. say  $p_j$, then $u$ would not buy the edge towards any vertex node $a_i$ such that $v_i$ is an endvertex of $e_j$. This is simply because the edge $(v,p_j)$ would only affect the distance between $u$ and $p_j$. Moreover, either $p_j'$ would be at distance 4 from $u$ or $u$ would have also bought the edge towards $p_j'$. In either case, $u$ would have convenience in deleting the edge towards $p_j$ -- as well as the edge towards $p_j'$, if she has bought it -- and in buying the edge towards a vertex node $a_i$, with $v_i$ being an endvertex of $e_j$, thus, decreasing her overall cost by at least 1. 

Now we show that for any improving move $S_u$ in which $u$ buys only edges towards vertex nodes, there is another improving move in which $u$ buys only edges towards vertex nodes that correspond to a vertex cover of the Vertex Cover instance. Indeed, if this is not the case, then there exist two nodes, say $p_j$ and $p_j'$, which are at distance 4 from $u$. Let $a_i$ be a vertex node such that $v_i$ is an endvertex of $e_j$. Clearly, the distance from $u$ to $a_i$ is 3. Therefore, by buying the edge towards $a_i$ the cost of $u$ would decrease by at least 1.

As a consequence, we can restrict the strategy space for agent $u$ only to improving moves that correspond to vertex covers of the Vertex Cover instance.   
Let $k'$ be the number of edges bought by $u$ in any strategy of the restricted strategy space for $u$, and let $N$ and $m$ be the number of vertices and edges of the Vertex Cover instance, respectively, The cost of $u$ is equal to $2k'+2k'+3(N-k')+6m=3N+6m+k'$. Since $N$ and $m$ are fixed, we observe that the cost of $u$ is minimized when $k'$ is minimized. Hence, any improving move for $u$ would define a vertex cover of size of at most $k-1$.
\end{proof}

\subsubsection{1-2-Graphs for $\alpha \leq 1$}
Here we study the \OTNCG with $\alpha \leq 1$. We prove that in this case a NE network always exists. In contrast to the corresponding result for the original NCG~\cite{Fab03} we do not prove this via a generic construction. Moreover, we provide a simple algorithm which computes a social optimum network in polynomial time and we provide tight bounds on the PoA. 

\paragraph{Existence}
In the following we prove an interesting connection between existence of a NE for the \OTNCG with $\alpha \leq 1$ and $k$-spanners. The weight of a $k$-spanner is the total sum of its edge weights. The following results are inspired by Lemma~\ref{lemma_OPT_spanner}.
\begin{lemma}\label{lm:3/2-spanner_is_GE}
Let $\frac{1}{2}\leq\alpha \leq 1$ and let $G$ be a $\frac{3}{2}$-spanner of minimum weight. Then $G$ contains all the edges of $H$ of weight 1 and has a diameter of at most 3. 
\end{lemma}
\begin{proof}
Let $(u,v)$ be any edge of $H$. Since $d_G(u,v) \leq 3/2\cdot d_H(u,v) \leq 3/2\cdot w(u,v)$ and all edge weights are in the set $\{1,2\}$, we have that $d_G(u,v) \leq 3$; furthermore, if $w(u,v)=1$, then $d_G(u,v)=1$, i.e., $(u,v)$ is contained in $G$. Therefore, $G$ contains all the edges of $H$ of weight 1 and has a diameter of at most 3.
\end{proof}

\begin{theorem}\label{thm_1-2_graph_alpha_1_NE}
Let  $\frac{1}{2}\leq\alpha \leq 1$ and let $G$ be a $\frac{3}{2}$-spanner of minimal weight. There is an edge ownership assignment in $G$ such that $G$ is in NE. 
\end{theorem}
\begin{proof}
The claim is proved by contradiction. Consider any edge ownership assignment in $G$ which induces strategy profile $\mathbf{s}$ and assume there is an agent $u\in V$ who can improve on her strategy $S_u$ in $\mathbf{s}$. We will show that if there is a better strategy $S'_u$ for agent $u$, then $|S'_u|\leq |S_u|-1$ and that $S_u'$ contains strictly less 2-edges than $S_u$. Then we prove that for any edge $(u,v)$, which would be removed by agent $u$ in the strategy change from $S_u$ to $S_u'$, we can exchange the ownership of its endpoint such that the new owner $v$ cannot improve on her strategy, or we can apply a combination of the two strategies $S'_u$ and $S'_v$ to $G$ which yields a new graph which is a $3/2$-spanner with less total weight, which contradicts that $G$ is a $3/2$-spanner of minimum weight. Therefore, the edge ownership can be chosen such that graph $G$ is in NE.

First, we prove that $|S'_u|\leq |S_u|-1$ and that $S_u'$ contains less 2-edges than $S_u$. Towards this we claim that the change from $S_u$ to $S_u'$ can only consist of a change of the 2-edges which are bought by $u$ and, if $\alpha = 1$, possibly the removal of some 1-edges. This is true since by Lemma~\ref{lm:3/2-spanner_is_GE} we have that all 1-edges are contained in $G$ and by Lemma~\ref{lem:1_edge_property} removing any 1-edge is not an improving move, in particular, removing a 1-edge is a cost neutral move if $\alpha = 1$. Using the latter, we can define a new strategy $S_u''$ which is identical to $S_u'$ but still has all the 1-edges which are contained in $S_u$. Thus, $S_u''\setminus S_u'$ only consists of 1-edges which are cost neutral for agent $u$ under strategy $S_u'$. Hence, $cost(u,S_u'') = cost(u,S_u')$ and we have $|S_u''| \geq |S_u'|$. Let $S_{u+}=\{v\in V: v\in S''_u\setminus S_u\}$ be the set of nodes to which new edges have been added, $S_{u-}=\{v\in V: v\in S_u\setminus S''_u\}$ be the set of nodes to which the edges have been deleted and let $G''$ be the graph obtained from $G$ by exchanging agent $u$'s strategy $S_u$ with~$S''_u$. 
Since the diameter of $G$ is 3, then, after changing the strategy from $S_u$ to $S''_u$, only distances between $u$ and nodes at hop-distance 2 from $u$ might increase. Thus, if there is a node $v\in S_u\setminus S''_u$ such that $d_{G''}(u,v)\geq 4$ or $d_{G''}(u,x)\geq 4$, where $x$ is at distance 1 from $v$, then the deletion of $v$ from strategy $S_u$ is not an improvement for $u$. This means that for any node $v\in V$ we have $d_{G''}(u,v)\leq 3$.  Therefore, the new strategy $S''_u$ decreases agent $u$'s edge cost by $2\alpha\cdot(|S_u|-|S''_u|)$, increases her distance to all nodes in $S_{u-}$ by 1 and decreases her distance by 1 for $|S_{u+}|$ many nodes. Since we assume that $cost(u,S_u)> cost(u, S''_u)$, then $0 > -2\alpha\cdot(|S_u|-|S''_u|) + |S_{u+}| - |S_{u-}|= (2\alpha+1)(|S_{u+}|-|S_{u-}|)$, thus, $|S_{u+}|\leq |S_{u-}|-1$. Hence, $|S_u''| \leq |S_u|-1$, i.e., $S_u''$ contains strictly less 2-edges than $S_u$. Since $|S_u''| \geq |S_u'|$, we have $|S_u'| \leq |S_u|-1$.

Let $G'$ be the graph obtained from $G$ by exchanging agent $u$'s strategy $S_u$ with strategy $S_u'$. Since the number of edges in $G'$ is strictly less than the number of edges in $G$ and since $G'$ has strictly less 2-edges than $G$, it follows that if the diameter of $G'$ is 3, then $G'$ is a $3/2$-spanner of total weight less than the total weight of $G$ and we get a contradiction. But it might happen that there are at least two nodes $x, y\in V $ at distance 4 in $G'$. Note that if the distance between $x$ and $y$ increased because of removing the edge $(u,v)$ ,i.e., $d(x,y) = d(x,u) + w(u,v) + d(v,y)$, then $w(u,v)=2$. Indeed, if $(u,v)$ was a 1-edge, then the distance between $u$ and $y$ as well as the distance between $u$ and $v$ would increase by 1. Therefore, the 1-edge $(u,v)$ would not a neutral edge and its removing is not an improving move, i.e., $v\in S'_u$. Hence, any edge whose deletion influences the distance between not only its endpoints must be a 2-edge. Since for any $v\in V$ we have $d_{G''}(u, v)\leq 3$ and since $G''$ and $G'$ only differ in 1-edges bought by agent $u$ whose removal increases the distance only to the other endpoint, it follows that for any $v\in V$ we have $d_{G'}(u, v)\leq 3$.
Since the diameter of $G$ is 3 and since for any $v\in V$ we have $d_{G'}(u, v)\leq 3$, then $x$ must be a neighbor of $u$ which is connected by a 1-edge and $y\in S_u\setminus S'_u$. For each such edge $(u,y)$ we can invert the ownership and, if none of the new owners has an improving strategy which does not contain $u$ then agent $u$ has a strategy she cannot improve on. 

Now we prove that after the inversion of the edge ownership for each edge $(u,y)$, for all $y\in S_u \setminus S_u'$, no agent $y$ can have an improving strategy which does not contain $u$. Assume towards a contradiction that $U$ is the non-empty set of nodes $y$, which have an improving strategy $S'_y$ which does not contain $u$. We apply all improving strategies $S'_y$, for all $y\in U$, and $S'_u$ to $G$ and obtain a new graph $G^*$. Note that if there are two nodes $x, y$ such that there is a 2-edge $(x,y)\in E(G)$, the edge can be removed by one of the endpoints, say $x$. This move does not influence the strategy of the agent $y$, since otherwise there must be a node $v\in V$, which is at distance 1 from $x$ and $d_G(y,v)=w(y,x)+w(x,v) = 2+1$, and then we could assign the ownership of $(x,y)$ to agent $y$ and then the edge $(x,y)$ would not be removed from $G$. Therefore, all the strategies can intersect only in pairs of nodes that want to add the same edge.   

Note that for any $y\in U$ we have $S_{y+}\cap S_{u+}=\emptyset$ and $S_{y-}\cap S_{u-}=\{(u,y)\}$, and for all $v\in V$ we have $d_{G^*}(u,v)\leq 3$ and $d_{G^*}(y,v)\leq 3$. The number of edges in $G$ is $|E(G)|=\bigl|\tilde E\cup\left(\bigcup_{y\in U}{S_{y-}}\right)\cup S_{u-}\bigr|= |\tilde E| + \sum_{y\in U}{|S_{y-}|} + |S_{u-}| - |U|$, where $\tilde E \subset E(G)$ is a set of edges which are both in $G$ and in $G^*$. On the other hand,  $|E(G^*)|=\bigl|\tilde E\cup\left(\bigcup_{y\in U}{S_{y+}}\right)\cup S_{u+}\bigr|\leq |\tilde E| + \sum_{y\in U}{|S_{y+}|} + |S_{u+}|\leq |\tilde E| + \sum_{y\in U}{\left(|S_{y-}|-1\right)} + |S_{u+}| - 1 =|\tilde E| + \sum_{y\in U}{|S_{y-}|} + |S_{u+}| - |U| - 1 < |E(G)|$. Hence, since only 2-edges were modified, the new graph $G^*$ is a $3/2$-spanner with less weight than the weight of the spanner $G$, which contradicts that $G$ is a $3/2$-spanner with minimum weight. Therefore, the edge ownership can be chosen such that the graph $G$ is in NE.
\end{proof}

\paragraph{Optimal networks} Now we consider how to compute a social optimum network.

\begin{algorithm2e}[h]
 \Input A complete graph $G = K_n$;
 
 \While{there is 1-1-2 triangle in $G$}{
 Remove the edge of weight 2 from the triangle;}
 	
	\caption{ computes a social optimum for the \OTNCG in polynomial time.}\label{alg:opt_1_2_graphs}
\end{algorithm2e}

\begin{theorem}\label{th:alg_produces_opt_for_1_2_graphs}
For any $\alpha\leq 1$, algorithm \ref{alg:opt_1_2_graphs} produces an optimal network in polynomial time.
\end{theorem}
\begin{proof}
Let $G^*$ be an optimal network. We first prove that there is an optimal network of diameter 2. We assume that $G^*$ has diameter strictly greater than 2. Let $u$ and $v$ be the vertices at distance greater than or equal to 3 in $G^*$. We show that $G^*+(u,v)$ is also an optimal network. Indeed, the cost of adding the edge to the network is at most $2\alpha \leq 2$, while the sum of the all-to-all distances decreases by at least 2 as the distance between $u$ and $v$ decreases by at least 1.

Next, we show that the social optimum contains all 1-edges. Indeed, if one 1-edge, say $(u,v)$, were missing in $G^*$, then $G^*+(u,v)$ would be a network which is cheaper than $G$, because its edge cost is at most 1 plus the edge cost of $G$, while its distance cost is at most the distance cost of $G$ minus 2.

Now, observe that the network $G$ produced by the algorithm has diameter equal to 2 and contains all 1-edges. The claim follows by observing that every network of diameter 2 that contains all the 1-edges has to contain all the edges of $G$.
\end{proof}

\paragraph{Price of Anarchy}

We start with the following technical lemma observing a relation between stable networks and the corresponding optimum.

\begin{lemma}\label{lem:E(NE)_is_subset_of_E(OPT)}
Consider $0<\alpha\leq 1$. Let $G^*$ be the social optimum obtained by Algorithm~\ref{alg:opt_1_2_graphs} and let $G$ be a stable network. Then $E(G)\subseteq E(G^*)$. Moreover, $d_G(u,v)=2$ for every 1-edge $(u,v)\notin~E(G)$ and $d_G(u,v)\leq 3$ for every 2-edge $(u,v)\notin E(G^*)$.
\end{lemma}
\begin{proof}
We observe that $G^*$ contains all the 1-edges and has diameter 2. So, every 1-edge contained in $G$ is also contained in $G^*$. Let $(u,v)$ be a 1-edge that is not contained in $G$. We have $d_G(u,v)=2$, as otherwise $u$ could buy the edge towards $v$ to improve her cost by at least 1.

Let $(u,v)$ be a 2-edge of $G$. We show by contradiction that $(u,v)$ is also contained in $G^*$. Assume that $(u,v)$ is not contained in $G^*$. Since $G^*$ has diameter 2, there exists a vertex $x$ such that $(u,x)$ and $(v,x)$ are two 1-edges. First of all, we observe that $G$ cannot contain both the edges $(u,x)$ and $(v,x)$ as otherwise the agent that is buying the 2-edge $(u,v)$ would remove such an edge without increasing any point-to-point distance in the graph and thus saving $2\alpha$ of her edge cost. We split the proof into two cases, according to whether exactly one of the two edges between $(u,x)$ and $(v,x)$ is contained in $G$, or not, and we show how to obtain a contradiction in either case.

We consider the case in which either $(u,x)$ or $(v,x)$ is an edge of $G$. W.l.o.g., we assume that $(u,x)$ is an edge of $G$. Since $d_G(v,x)=2$, there is a vertex, say $y$, such that $(x,y)$ and $(y,v)$ are two 1-edges of $G$. If the edge $(u,v)$ is bought by $v$, then $v$ can improve her cost by swapping the edge $(u,v)$ with the edge $(x,v)$. By this the edge cost decreases by $\alpha$ and no distances from $v$ towards all the other vertices increases. Therefore, the edge $(u,v)$ is bought by player $u$. Because $G$ is stable, there is a vertex $z$ such that the unique shortest path from $u$ to $z$ passes through $v$, as otherwise $u$ would never have bought the edge towards $v$. Therefore, we have $1+d_G(x,z)=d_G(u,z)+d_G(x,z)\geq d_G(u,v)+d_G(v,z)+1=d_G(v,z)+3$ which implies that $d_G(x,z) \geq d_G(v,z)+2$. But in this case, $x$ can improve on her cost by buying the 1-edge towards $v$. By this, her edge cost increases by at most 1 while both the distances towards $v$ and $z$ decrease by at least 1. Hence, $G$ could not be stable.

We consider the case in which neither $(u,x)$ nor $(v,x)$ is an edge of $G$. Since $d_G(u,x)=d_G(v,x)=2$, there are two vertices, say $y$ and $z$, such that $(u,y),(y,x),(x,z)$, and $(z,v)$ are four 1-edges in $G$. We claim that $d_G(u,z)=d_G(v,y)=2$. We prove the claim for $d_G(u,z)$ as the proof for $d_G(v,y)$ uses similar arguments. The claim is proved by contradiction. If $d_G(u,z)=1$, then the player buying the edge $(u,v)$ may remove such an edge, without increasing any point-to-point distance, and thus saving a cost of $2\alpha$. If $d_G(u,z)\geq 3$, then $u$ can improve on her cost by buying the 1-edge towards $x$. By this, the edge cost increases by at most 1 while both the distances towards $x$ and $z$ decrease by 1. As a consequence, there is a vertex $w$ such that the unique shortest path from $u$ to $w$ in $G$ passes through $v$, as otherwise $u$ would never bought the edge towards $v$. Therefore, we have $2+d_G(x,w)=d_G(u,x)+d_G(x,w)\geq d_G(u,v)+d_G(v,w)+1=d_G(v,z)+3$ which implies that $d_G(x,w) \geq d_G(v,z)+1$. But in this case $x$ can improve on her cost by buying the 1-edge towards $v$. Indeed, by this the edge cost increases by at most 1 while both the distances towards $v$ and $w$ decrease by at least 1. Hence, $G$ could not be stable.

To complete the proof, it remains to show that $d_G(u,v)\leq 3$ for every 2-edge $(u,v)$ that is not in $G^*$. Let $(u,v)$ be a 2-edge that is not in $G^*$. Since $E(G)\subseteq E(G^*)$, $(u,v)$ is not contained in $G$. We prove by contradiction that $d_G(u,v) \leq 3$. For the sake of contradiction, assume that $d_G(u,v)\geq 4$.  Since $G^*$ has diameter 2, there is a vertex $x$ such that $(u,x)$ and $(v,x)$ are two 1-edges. Since $d_G(u,x),d_G(v,x) \leq 2$, both edges $(u,x)$ and $(x,v)$ are missing from $G$. Furthermore, there are two vertices, say $y$ and $z$, such that $(u,y),(y,x),(x,z)$, and $(z,v)$ are 1-edges in $G$. Because $d_G(u,v)\geq 4$, we have that $d_G(u,z)\geq 3$. In this case, $u$ can improve her cost by buying the edge towards $x$. By this, the edge cost increases by $\alpha\leq 1$ while all the distances towards $x$, $z$, and $v$ decrease by 1.
\end{proof}

\begin{theorem}\label{th:1_2_graph_UB_PoA_alpha_between_one_half_and_1}
For $1/2 \leq \alpha < 1$, $PoA \leq 3/(\alpha + 2)$.
\end{theorem}
\begin{proof}
First of all, we observe that both the social optimum and any NE contain all the 1-edges. 

Let $G$ be a NE and let $u$ and $v$ be two distinct vertices. 
Let $x$ and $x^*$ be two Boolean variables such that $x = 1$ iff $(u, v)$ is an
edge of $G$ and $x^* = 1$ iff $(u, v)$ is an edge of the social optimum $\OPT$.
We prove the claim by showing that
$$
\sigma:= \frac{\alpha \cdot w(u,v)x+2d_G(u,v)}{\alpha \cdot w(u,v)x^*+2d_{\OPT}(u,v)} \leq \frac{3}{\alpha+2}.
$$
First of all, we observe that if $w(u,v)=1$, then $x=x^*=1$. Furthermore, if $x=x^*=1$, then $\sigma = 1$. Therefore, we only need to prove the claim for the case in which $w(u,v)=2$ and $x$ and $x^*$ that cannot be both equal to 1. 

Let $G'$ be the graph induced by all the 1-edges. We observe that if $d_{G'}(u,v)=2$, then neither $\OPT$ nor $G$ contains the edge $(u,v)$ since $G'$ is a subgraph of both $\OPT$ and $G$. Therefore, we assume that $d_{G'}(u,v)\geq 3$. In this case, we have that $x^*=1$: indeed, the addition of edge $(u,v)$ to $\OPT$ would increase the edge cost by $2\alpha$, but would decrease the overall sum of all-to-all distances by at least 2. Similarly, if $d_{G'}(u,v)\geq 4$, then $x=1$. Since we are considering the case in which $x$ and $x^*$ cannot be both equal to 1, but $x^*=1$, it follows that $x=0$. Therefore, $d_{G'}(u,v)=3$ and thus, $d_{G}(u,v) = 3$. Hence, $\sigma \leq 6/(2\alpha+4) = 3/(\alpha+2)$. The claim follows.
\end{proof}

\noindent We proceed with a lower bound on the PoA which matches the upper bounds given in Theorem~\ref{thm:metric_PoA} and Theorem~\ref{th:1_2_graph_UB_PoA_alpha_between_one_half_and_1}.
\begin{theorem}\label{thm:LB_PoA_1_2_graph_alpha_1}
For every constant $\epsilon > 0$, 
\begin{equation*}
PoA \geq 	\begin{cases}
				3/2 - \varepsilon				&	\text{if $\alpha=1$;}\\
				3/(\alpha+2) - \varepsilon		&	\text{if $1/2 \leq \alpha < 1$}.
			\end{cases}
\end{equation*}
\end{theorem}
\begin{proof}
We prove the lower bound for $\alpha=1$ first. Consider the host graph $H$ contains a clique $K$ of $N$ vertices formed by 1-edges only. Each vertex $v$ of the clique is the center of star $X_v$ made of 1-edges only and whose leaves are $N$ new vertices. Finally, there is a new vertex, that we call $u$, that is connected to every other vertex by a 1-edge. Thus, the overall number of vertices of the host graph is $n=N^2+N+1$. All other edges are 2-edges.
\begin{figure}[ht]
\begin{minipage}[ht]{0.5\textwidth}
\center{\includegraphics[width=0.5\textwidth]{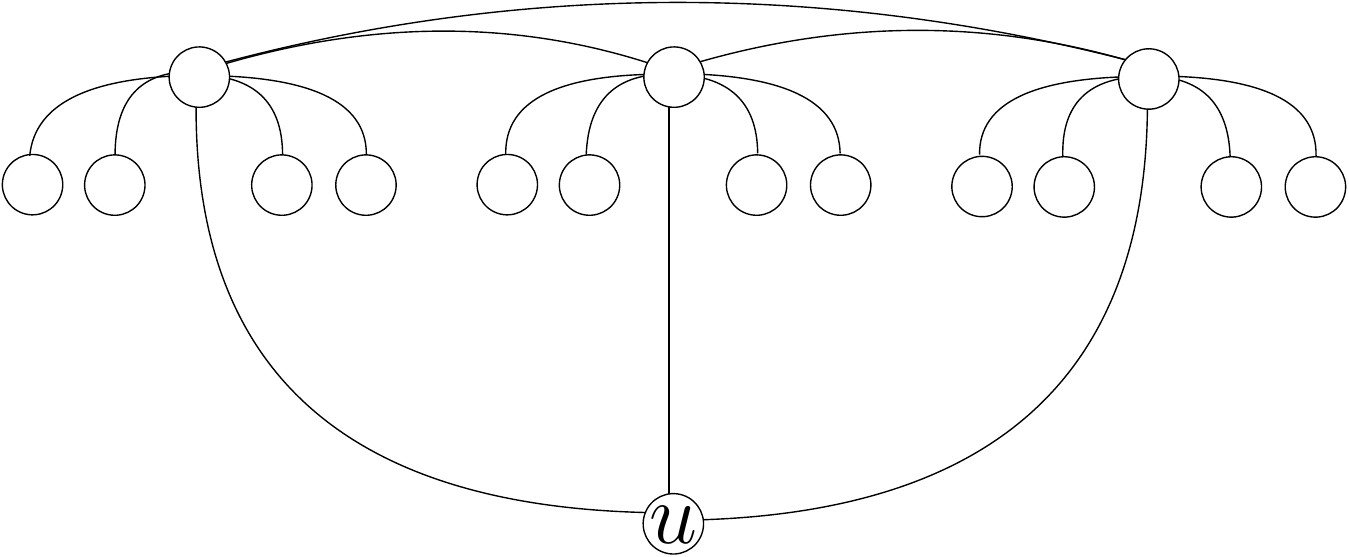}}
\end{minipage}
~
\begin{minipage}[ht]{0.5\textwidth}
\center{\includegraphics[width=0.55\textwidth]{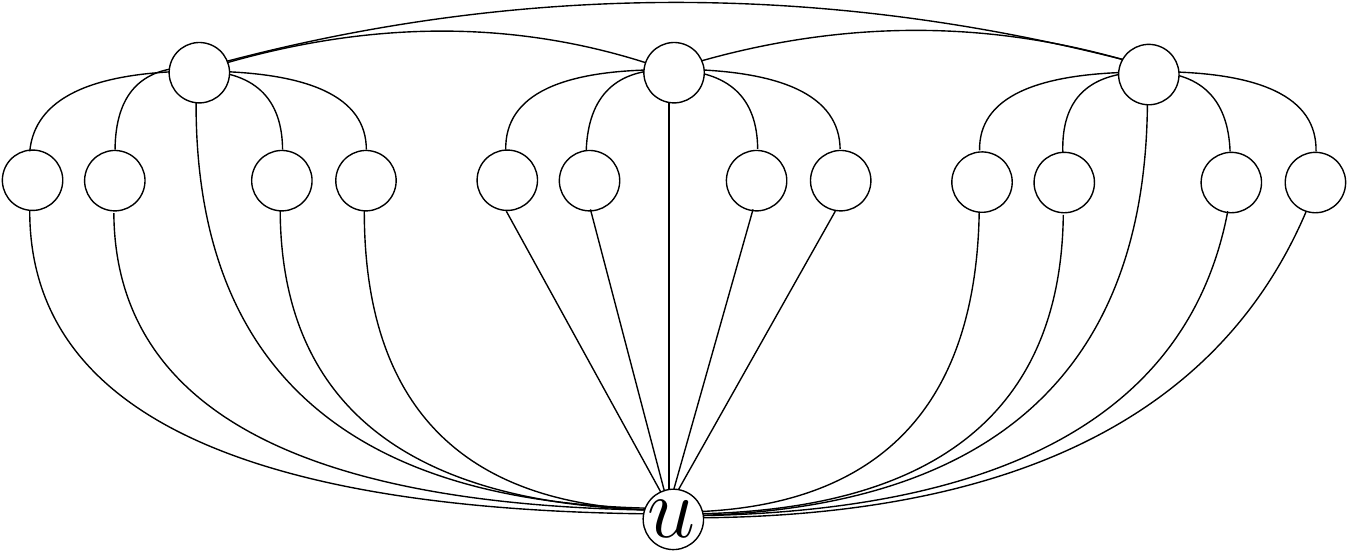}}
\end{minipage}
\caption{NE and optimal networks for $N=4$. All depicted edges are 1-edges. On the right hand side, the optimal network for $\alpha=1$ is depicted. The graph on the left hand side is a NE for every $1/2 \leq \alpha \leq 1$.}\label{fig:LB_PoA_example}
\end{figure}

We observe that the social optimum corresponds to exactly the subgraph induced by the 1-edges. Therefore, the edge cost of the social optimum is $O(N^2)$, while the distance cost is at most $2N^4+2N^2$. Therefore, the social cost of the social optimum is at most $2N^4+O(N^2)$.

We claim that the subgraph induced by all the 1-edges, except for those among $u$ and the leaves of each of the $X_v$, is a NE. Indeed, since the resulting graph has diameter 3, no player has an incentive to buy a 2-edge. Furthermore, no player has an incentive in removing a 1-edge. Finally, neither $u$ nor any leaf of any star $X_v$ has an incentive in buying the leftover 1-edge connecting them. The edge cost of the stable network is $O(N^2)$, while the distance cost is at least $3N^2(N(N-1))=3N^4-3N^3$, since any leaf of any star $X_v$ has all the $N(N-1)$ leaves of the other stars at distance 3. Therefore, the social cost of the stable graph is $3N^4-\Theta(N^3)$. The claim for $\alpha=1$ follows by choosing a sufficiently large value of $N$.

Now, we prove the claim for $1/2 \leq \alpha < 1$. Let the graph on the left hand side of Figure \ref{fig:LB_PoA_example} be the host graph (only 1-edges are depicted). A trivial upper bound on the social optimum is the cost of the entire host graph. Therefore, the edge cost of the social optimum is upper bounded by $2\alpha\frac{(N^2+N+1)(N^2+N)}{2} = \alpha N^4 + \Theta(N^3)$, while the distance cost is upper bounded by $2(N^2+N+1)(N^2+N) = 2N^4 + \Theta(N^3)$. Therefore, the cost of the social optimum is $(\alpha + 2)N^4 + \Theta(N^3)$. Once again, the subgraph induced by all the 1-edges is a NE as, for $\alpha < 1$, any NE contains all the 1-edges. Furthermore, since the resulting graph has diameter 3, no player has an incentive to buy a 2-edge as $\alpha \geq 1/2$. As already proved for the case $\alpha=1$, the social cost of the NE graph is $3N^4-\Theta(N^3)$. The claim now follows by choosing a sufficiently large value of $N$.
\end{proof}

Now we show that selfishness does not lead to a loss in social welfare if $\alpha$ is small enough.

\begin{theorem}\label{th:1_2_graph_PoA_is_1}
For any $0\leq\alpha < \frac{1}{2}$ the PoA is equal to $1$.
\end{theorem}
\begin{proof}
We claim that, if $0\leq\alpha < \frac{1}{2}$, any NE network is equal to the optimal network produced by Algorithm~\ref{alg:opt_1_2_graphs}. To show this we need to prove that a NE network does not contain 1-1-2 triangles but containes all the other edges. 

Consider a graph which is in NE. It is easy to see that all 1-edges are contained in $G$ because otherwise, the addition of any such edge improves distance cost by at least 1 and increases the edge cost by $\alpha<1$. Moreover, if there is a 2-edge $(u,v)$, which is not in $G$ and does not form a 1-1-2 triage in $G+(u,v)$, then the addition of $(u,v)$ improves the distance cost by $d_G(u,v)-w(u,v)\geq 3 - 2 = 1$ and increases the edge cost by $2\alpha<1$, i.e., it implies an improvement for an owner of the edge. Finally, it is clear that $G$ does not contain 1-1-2 triangles because removing a 2-edge from such triangle does not change the distance cost. Therefore, $G$ is equal to the social optimum obtained by Algorithm~\ref{alg:opt_1_2_graphs}. 
\end{proof}

\subsubsection{1-2-Graphs for $\alpha > 1$}
In this section we show that the \OTNCG for $\alpha > 1$ behaves very similar to the original NCG. 

\begin{theorem}\label{thm:1_2_graph_star_is_NE}
For $\alpha\geq 3$ any star graph is in NE.
\end{theorem}
\begin{proof}
Consider a star graph $S_n$ that has $n-2$ edges. Assume that the central node $u$ is an owner of all edges in the star. Then $u$ cannot improve her strategy.  Let $v, z$ be two leaf nodes. The only possible strategy improvement for a leaf node is an edge addition. 
In the worth case $w(v,u)=w(z,u)=2$ and $w(z,v)=1$, thus, adding an edge $(z,v)$ improves the distance only between the edge endpoints by $3$ and costs $\alpha\geq 3$. Therefore, there is no strategy improvement for any agent. It implies that $S_n$ is in NE. 
\end{proof}

\paragraph{Price of Anarchy}
We use the proof technique from \cite{Fab03} to show that the PoA may be bounded by the same value as in the original proof. 
We start with the bounding the social cost of the~NE.

\begin{lemma}\label{lm:PoA_is_diameter}
Consider any NE $G$ in the \OTNCG. If $G$ has diameter $D$, then its social cost is at most $O(D)$ times the social cost of the optimal network $OPT$.
\end{lemma}
\begin{proof}
First, we evaluate the social cost of the optimal network. Since the network should be connected and each edge has length at least 1, the total cost is in $\Omega(\alpha\cdot n + n^2)$.

Now we analyze the social cost of $G=(V,E)$ which is in NE. The distance cost is trivially in $O(Dn^2)$ since each pair of vertices is in distance $D$ in $G$. To evaluate the edge cost we consider cut edges, whose removal disconnects $G$. There are at most $n-1$ cut edges in the graph, thus, the edge cost is at most $O(\alpha(n-1))$ plus the edge cost of noncut edges. 
Now consider a node $v$ in $G$ which has at least one noncut edge. We claim that the number of noncut edges of $v$ is at most $n(2D+1)/\alpha$, thus, the total edge cost of all noncut edges in $G$ is at most $2n^2(2D+1)$, which implies that $cost(G)\in O(\alpha(n-1)+2n^2(2D+1)+Dn^2)=O(\alpha n + Dn^2)$. And therefore we conclude that the ratio between $cost(G)$ and the cost of the optimal solution is in $O(D)$.

Consider an edge $e=(u,v)\in E(G)$, owned by the player $u$. Let $V_{e}$ be the set of nodes $w$, such that the edge $e$ is in the shortest path from $u$ to $w$. Let $G'$ be the graph $G$ without the edge $e$. Since $G$ is stable, we have $0\leq cost_{G'}(u)-cost_G(u)\leq -\alpha\cdot w(u,v) + (d_{G'}(u,v)-w(u,v))\cdot|V_e|$. We claim that $d_{G'}\leq 2d$. Indeed, consider a cycle which consists of the shortest path from $u$ to $v$ in $G'$ and the edge $(u,v)$. Let $v'\in V_{e}$ be the node which is furthest away from node $v$ in the cycle and let $(u',v')$ be its incident edge in the cycle but not in $V_{e}$. Since $u'\notin V_{uv}$, we have $d_{G'}(u,u')=d_G(u,u')\leq d$, and since $v'\in V_e$ we get $d_{G'}(v,v')=d_G(v,v')\leq d$. Thus, $d_G(u,v) \leq d_{G'}(u,u')+ w(u',v') + d_{G'}(v',v)\leq 2D+2$. Therefore, we get  $0\leq -\alpha\cdot w(u,v) + (d_{G'}(u,v)-w(u,v))\cdot|V_e|\leq -\alpha + (2D+2-1)|V_e|$. It follows that $|V_e|\geq \alpha/(2D+1)$. Thus, the total number of noncut edges of $v$ in $G$ is at most $n(2D+1)/\alpha$. This completes the proof.
\end{proof}

\noindent With the above lemma we can easily get the following.
\begin{theorem}\label{thm:1_2_graph_UB_PoA_sqrt_alpha}
The \OTNCG with $\alpha >1$ has a $PoA\in O(\sqrt{\alpha})$.
\end{theorem}
\begin{proof}
Using Lemma~\ref{lm:PoA_is_diameter}, we only need to prove that the diameter of the NE is at most $\sqrt{\alpha}$. We consider a pair of nodes $u,v$ in the graph $G$, which is in NE and has diameter $D$. Assume that $d_G(u,v)=D$. Since $G$ is in NE, the addition of the edge $(u,v)$ does not yield an improvement for agent $u$. Thus, $0\leq cost_{G+(u,v)}(u)-cost_G(u)$. Let $P:= v=v_1, v_2,\ldots,v_{m-1},v_m=u$ be the shortest $u-v$ path in $G$ and let $k=D/5$. We observe that the distances from $u$ to $v_1,\dots,v_k$ will all change after the addition of the edge $(u,v)$. Thus, taking into account that each edge has length at most $2$, we have:
$0\leq cost_{G+(u,v)}(u)-cost_G(v) \leq \alpha \cdot w(u,v) + \sum\limits_{i=1}^{k}\big(w(u,v)+d_{G+(u,v)}(v,v_i)-d_G(u,v_i)\big)\leq 2\alpha + \sum\limits_{i=1}^{k}\big(2i - (D - 2k)\big)
\leq 2\alpha + \sum\limits_{i=1}^{k}(4k - D) \leq 2\alpha  - \frac{D^2}{25}$. It follows that $D\in O(\sqrt{\alpha})$.
\end{proof}

We are convinced that also other proof techniques from the NCG can be carried over to the \OTNCG. Thus, the PoA should be constant for almost all $\alpha$ and in $o(n^\varepsilon)$ for the remaining range.

\subsection{Tree Metrics}\label{subsec_tree}
This section is devoted to the study of tree metrics. We assume that the host graph $H=(V,E)$ is defined as the {\em metric closure} of an edge-weighted tree $T$.\footnote{More precisely, $w(u,v) = d_T(u,v)$ for every two vertices $u$ and $v$.}

\paragraph{Existence}
The first result is about the structure of any NE. Differently for general metrics and 1-2-graphs, we prove that any NE in \TNCG is as much sparse as possible.
\begin{theorem}\label{thm_tree_NE_is_tree}
In the \TNCG any NE is a tree.
\end{theorem}
\begin{proof}
Consider a graph $G=(V,E)$ which is in NE. For the sake of contradiction, we assume that $G$ contains a cycle. Clearly, this cycle has at least one edge, say $(u,v)$, which is not contained in the tree $T$. Without loss of generality, assume $u$ be the owner of the edge $(u,v)$. Consider a vertex $x\in V$ such that the edge $(x,v)$ is in the shortest $u$-$v$--path in $T$. Note that $(x,v)\notin E(G)$, otherwise swapping the edge $(u,v)$ to $(u,x)$ is an improving move that contradicts with $G$ being in NE. Consider two possible situations: $d_G(u,x)>w(u,x)$ and $d_G(u,x) = w(u,x)$. 

In case $d_G(u,x)>w(u,x)$, consider a graph $G'=(V,E')$ obtained from $G$ by swapping the edge $(u,v)$ to $(u,x)$ by agent $u$. Denote by $Z=\{z\in V: d_G(u,z)<d_{G'}(u,z)\}$ the set of vertices to which the distance from $u$ has increased. Note that $Z$ is not an empty set because $v\in Z$. Since $G$ is in NE, $\alpha \cdot w(u,x) + w(u,x) + d_{G'}(u,Z)\geq \alpha \cdot w(u,v) + d_G(u,x) + d_G(u,Z)$, whereas the left part of the inequality is at most $\alpha \cdot w(u,x)+ w(u,x) + |Z|\cdot w(u,x) + d_G(x,Z)$ and the right part is equal to $\alpha \cdot w(u,v) + d_G(u,x) +  |Z|\cdot w(u,v) + d_G(v,Z)$. Since $d_G(u,x)>w(u,x)$, we get $\alpha \cdot w(u,x)+ |Z|\cdot w(u,x) + d_G(x,Z) > \alpha \cdot w(u,v) + |Z|\cdot w(u,v) + d_G(v,Z)$.

In case $d_G(u,x)=w(u,x)$, consider deletion of the edge $(u,v)$. Since $(u,v)$ is in the cycle, the deletion does not increase any distance in $G'$ to infinity. We use the same notation: $G'=(V,E')$ is a graph after modification, $Z$ is a set of verticies to which the distance from $u$ has increased. As before, $Z\neq\emptyset$ because otherwise deletion of the edge $(u,v)$ is an improving move for agent $u$. Since $G$ is stable, $|Z|\cdot w(u,x) + d_G(x,Z)\geq d_{G'}(u,Z)\geq \alpha \cdot w(u,v) + d_G(u,Z) = \alpha w(u,v) + |Z|\cdot w(u,v) + d_G(v,Z)$. Adding positive term $\alpha \cdot w(u,x)$ to the left part of the inequality, we get the same inequality as in the previous case: $\alpha \cdot w(u,x)+ |Z|\cdot w(u,x) + d_G(x,Z) > \alpha \cdot w(u,v) + |Z|\cdot w(u,v) + d_G(v,Z)$. 

From the other hand, the agent $x$ does not buy the edge $(x,v)$.  Therefore, $\alpha \cdot  w(x,v) + |Z|\cdot w(x,v) + d_G(v,Z) \geq d_G(x,Z)$. We sum up this inequality with the inequality we obtained by analyzing swapping and deletion and we get $\alpha (w(x,v) + w(u,x)) + |Z|(w(x,v) + w(u,x)) + d_G(x,Z) + d_G(v,Z) > \alpha \cdot w(u,v) + |Z|\cdot w(u,v) + d_G(v,Z) + d_G(x,Z)$. Simplifying and taking into account that $w(u,x) + w(x,v) = w(u,v)$, we get that $(\alpha + |Z|)\cdot w(u,v) > (\alpha + |Z|)\cdot w(u,v)$, that is a contradiction. Therefore, $G$ has no cycles, i.e., $G$ is a tree.  
\end{proof}

The next result follows by observing that the tree $T$ that defines the metric is the network with cheapest total edge cost that preserves all the distances of the host graph at the same time.

\begin{corollary}\label{thm:T_is_OPT}
In the \TNCG the tree $T$ which defines the metric is both the social optimum and in NE.
\end{corollary}

\noindent Corollary~\ref{thm:T_is_OPT} yields that the cheapest NE is also a social optimum.\footnote{In Algorithmic Game Theory, this is equivalent to say that the {\em Price of Stability} -- the ratio between the cost of the cheapest NE and the cost of a social optimum -- is 1.}

\paragraph{Hardness}

We prove that the problem of computing the best response of a player is NP-hard for tree metrics (\TNCG).

\begin{theorem}\label{thm_tree_hardness}
It is NP-hard to compute a best response of an agent in the \TNCG. 
\end{theorem}
\begin{proof}
We perform the proof by a reduction from the Minimum Set Cover problem, which is well-known to be NP-hard. The problem is defined as follows: for a given universe $U=\{1,2,\ldots,k\}$ and a collection of non-empty subsets $\mathcal{X}=\{X_1,\ldots,X_m\}$ such that for any $1\leq i \leq m$ we have $X_i\subseteq U$ and $\bigcup\limits_{i=1}^m {X_i} = U$. It is required to find minimum number of subsets covering~$U$.

We define the corresponding instance of the best response problem in the \TNCG with $\alpha=1$ as follows: Consider a tree $T = (V, E_T)$ which defines metric such that $$V = \{u, c\}\cup \{a_1,\ldots,a_m\}\cup\{b_1,\ldots,b_m\}\cup \{p_1,\ldots,p_k\}$$ and $$E=\{(c,u)\}\cup\bigcup\limits_{i=1}^m\left(\{(b_i,u), (c,a_i)\}\cup\{(a_i,p_j)|\ p_j\in X_i\}\right),$$ where each $p_j$ represents one element of the universe $U$ and each $a_i$ corresponds to one subset $X_i$. All nodes $c, b_1,\ldots, b_m$ are connected with $u$ and each edge $(u,b_i)$ has length $\frac{1}{2}(L-\beta)$, whereas the edge $(u,c)$ is of length $L-\varepsilon$. Each of set nodes $a_1, \ldots, a_m$ is connected with $c$ by an edge of length $\varepsilon$. Furthermore, all edges between the element nodes $p_1, \ldots, p_m$ and the set nodes are of length $L$ and each set element node connected with only one set node.  We assume throughout the proof that $L >>\varepsilon$ holds. Moreover, we assume that each edge $(b_i,u)$ is owned by the respective node $b_i$. Finally, note that agent $u$ does not own any edges in $G$. 
See the right side Figure~\ref{fig:NP_tree_metric} for the illustration of the constructed graph.

Consider a graph $G=(V,E)$ such that each node $c, b_1, \ldots, b_m$ is an owner of the edge connecting it with the node $u$. For all $i=1,\ldots,m$ there is an edge $(b_i,a_i)$ of length $\frac{1}{2}(L-\beta)+(L-\varepsilon)+\varepsilon =\frac{1}{2}(L-\beta)+L$. Also each element node $p_j$ is connected with some set node $a_i$ iff the element is in the set corresponding to the set node. Note that $L\leq w(a_i,p_j)\leq L+2\varepsilon$.  See figure~\ref{fig:NP_tree_metric}(left) for the illustration of the constructed graph $G$ and the metric tree $T$.

\begin{figure}[h!]
\center{\includegraphics[width=0.9\textwidth]{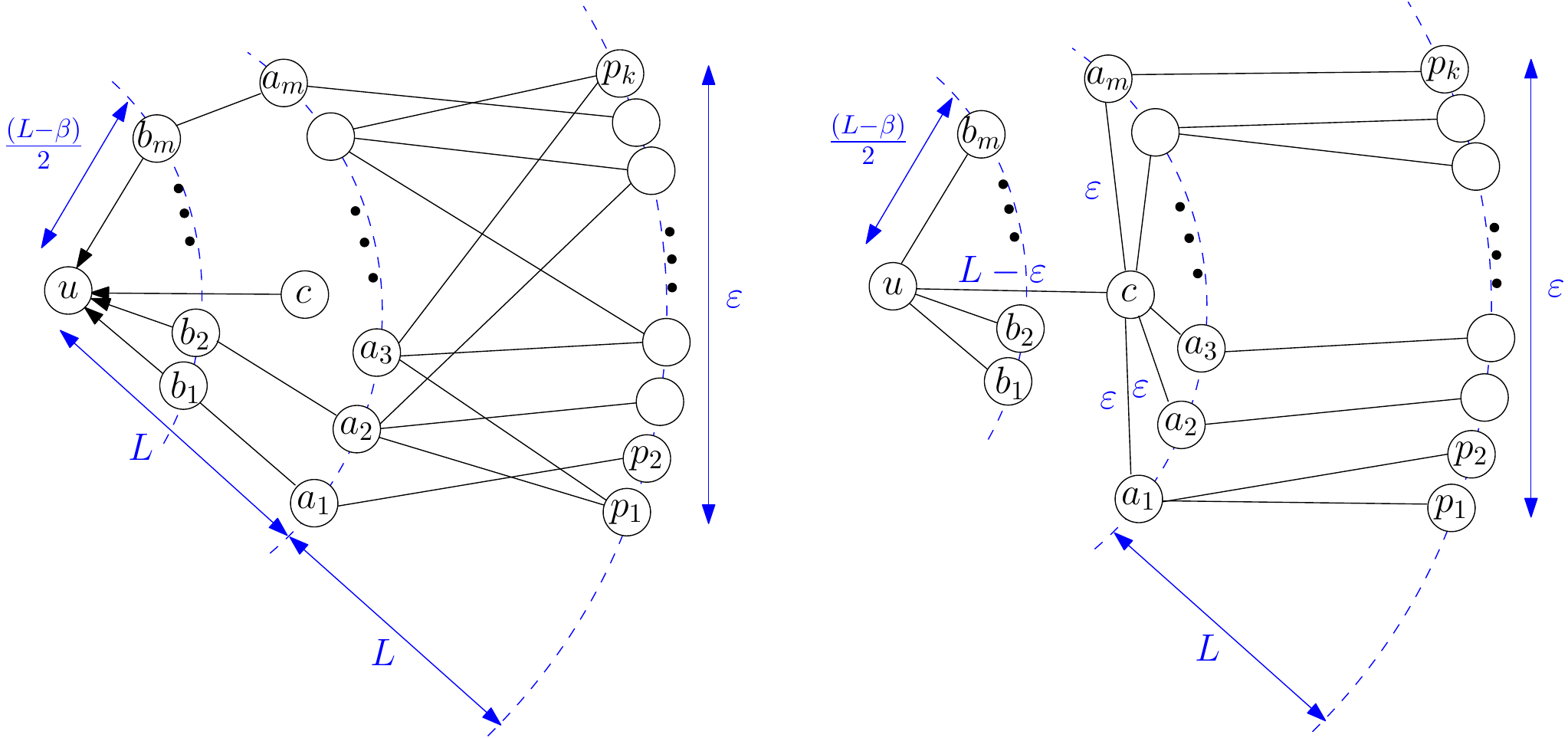}}
\caption{Illustration of the construction used in the reduction. The left figure shows the constructed graph $G$. On the right figure the given metric tree $T$ is shown .}\label{fig:NP_tree_metric}
\end{figure} 

We claim that the best response $S_u^*$ of the player $u$ contains only the set nodes corresponding to a solution of the Minimum Set Cover problem. 

First, we prove that $S_u^*$ includes only set nodes. Assume that there is a node $p_j\in S_u^*$. Consider a node $a_i$ corresponding to the set $X_i$ containing element $p_j$. If $a_i\in S_u^*$, then the player $u$ can delete the edge $(u,p_j)$ because this move decreases her edge cost by $w(u,p_j) = 2L$ and increases her distance cost by at most $2\varepsilon$, thus, her total cost decreases by at least $L-2\varepsilon$. If $a_i\notin S_u^*$, then swap of the edge $(u,p_j)$ to $(u,a_i)$ improves agent's cost by at least $L-2\epsilon-\beta$ because it decreases her edge cost by $L$, decreases distance cost to at least one node, which is $a_i$, by $L-\beta$ and increases distance only to the node $p_j$ by $2\varepsilon$.  

Next we show that all the nodes in $S_u^*$ correspond to a set cover. Indeed, if there is a node $p_j$ such that none of its incident set nodes is in $S_u^*$, then buying an edge to one of the corresponding set node $a_i$ decreases distance between the node $u$ and both nodes $a_i, p_j$ by at least $2(2L-\beta -L)=2(L-\beta)$. Hence, the social cost of the agent $u$ is improves by at least $L-2\beta$.  

Finally, we show that $S_u^*$ corresponds to a minimum set cover. Consider two strategies $S_u^1, S_u^2$ corresponding to two different set covers of all elements. Assume $\Delta := |S_u^2|-|S_u^1|>0$. Thus, the difference in the agent's $u$ cost with strategy $S_u^1$ compared with the strategy $S_u^2$  is $-\Delta\cdot L + (L-\beta)\Delta + 2k\varepsilon = -\Delta\beta +2k\varepsilon<0$. 
\end{proof}

\paragraph{Dynamic Properties} The following theorem shows that the network dynamics consisting of best responses only may never converge to a NE for \TNCG and thus also for \MNCG.
\begin{theorem}\label{thm_tree_no_FIP}
The \TNCG is not a potential game.
\end{theorem}
\begin{proof}[Proof(sketch)]

Consider the weighted tree depicted in Fig.~\ref{fig:BRC_tree_metric_combined}~(left). With this tree defining the metric distances, we can construct a best response cycle of length 4. See Fig.~\ref{fig:BRC_tree_metric_combined}~(right)).
\begin{figure}[h]
\center{\includegraphics[width=0.8\textwidth]{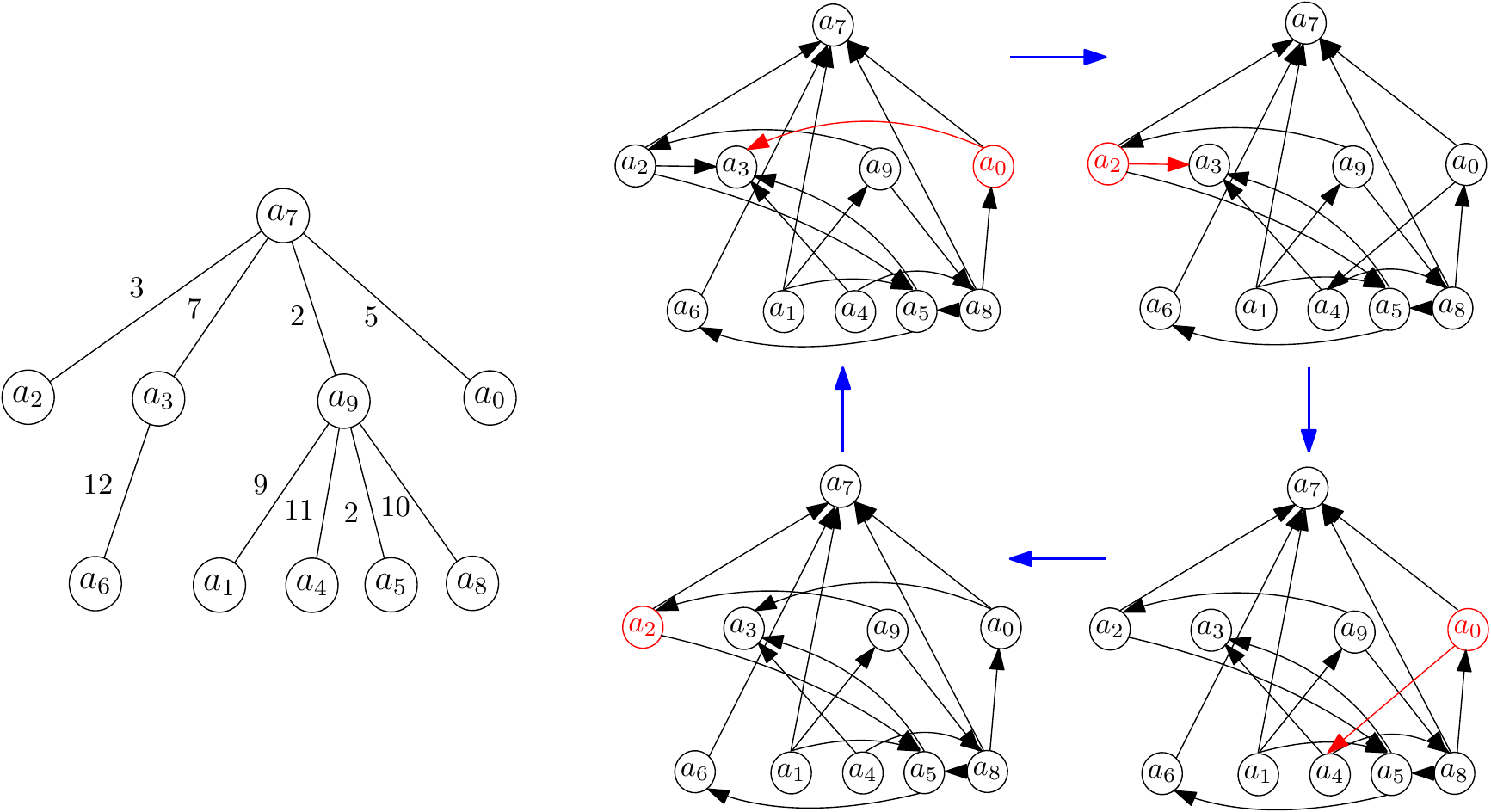}}
\caption{Left: Tree which defines the metric. Right: Best response cycle. The active agent is colored red. The direction of the edges denotes the owner.}\label{fig:BRC_tree_metric_combined}
\end{figure} 
\end{proof}

\paragraph{Price of Anarchy} We prove that the upper bound given in Theorem~\ref{thm:metric_PoA} for more general metric instances is tight for tree metrics and thus also for graph metrics.
\begin{theorem} \label{thm:LB_PoA_tree_metric}
The PoA in the \TNCG is at least $\frac{\alpha+2}{2}-\varepsilon$, for any $\varepsilon > 0$.
\end{theorem}
\begin{proof}
Let $S^*_n$ be the weighted tree which defines the metric distances. The tree $S^*_n$ is a star and contains $n-2$ edges of weight $2/\alpha$ and one edge $(u,v)$ of weight 1, where $u$ is the center of the star. See Fig.~\ref{fig:LB_PoA_tree_metric}~(left).
\begin{figure}[h]
\center{\includegraphics[width=0.6\textwidth]{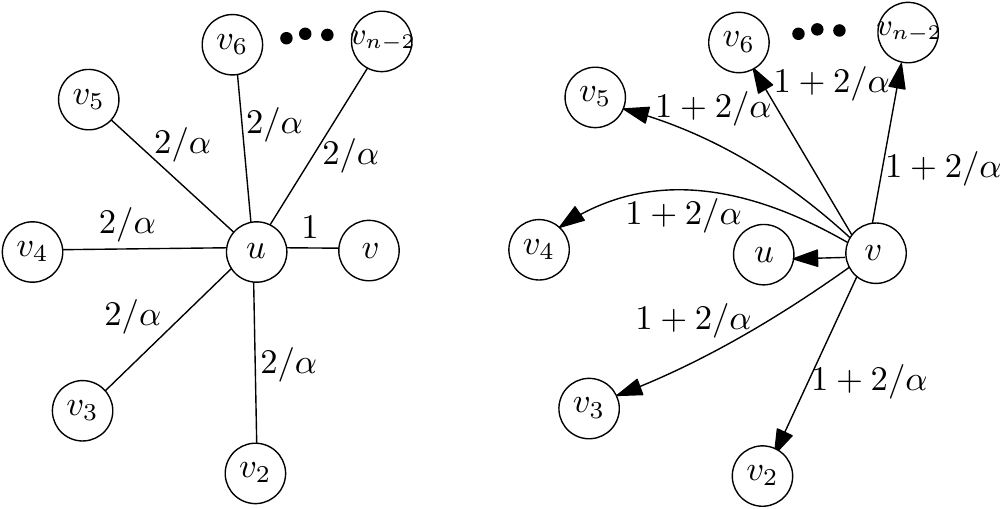}}
\caption{Left hand side figure shows a tree $T$, on the right figure is shown the stable graph $S_n$.}\label{fig:LB_PoA_tree_metric}
\end{figure}  

\noindent The star $S_n^*$ minimizes the social cost, which is 
\begin{align*}
 cost(S^*_n)=&\alpha \cdot edge(S^*_n)+ edge(S^*_n) + (n-2)(edge(S^*_n)+(n-2)\cdot 2/\alpha)+(edge(S^*_n)+n-2)\\=&(2n+\alpha-2)\cdot edge(S^*_n)= (2n+\alpha-2)\cdot \left((n-2)2/\alpha + 1\right).
\end{align*}


\noindent Let $S_n$ be a spanning star of the host graph such that a center of $S_n$ is the vertex $v$. The star $S_n$ contains one edge of weight 1 and $(n-2)$ edges of weight $(1+2/\alpha)$. Moreover, we assume that the central vertex $v$ is an owner of all edges in $S_n$. See Fig.~\ref{fig:LB_PoA_tree_metric}~(right).

We claim that $S_n$ is in NE. Indeed, the central agent $v$ cannot improve her strategy because all other vertices are leaves. No leaf owns an edge and, therefore, a leaf agent can possibly improve her strategy only by adding edges to other leaves. For any leaf agent $x\neq u$ buying the edge $(x,u)$ costs $\alpha\cdot 2/\alpha=2$ and improves her distances only towards $u$ by $2+2/\alpha - 2/\alpha=2$. Hence, buying $(x,u)$ is not an improvement. Buying any other edge costs $\alpha\cdot 4/\alpha=4$ for agent $x$ and improves her distance only to the endpoint of the new edge by $2+4/\alpha-4/\alpha=2$. At the same time, the agent $u$ cannot improve her strategy by buying edges to the leafs because it improves distance to each $v_i$ by 2 and increases edge cost by the same value for each new edge. 
Hence, no agent has an improving strategy change and it follows that $S_n$ is in NE. The social cost of $S_n$ is 
\begin{align*}
 cost(S_n) &=(2n+\alpha-2)\cdot edge(S_n)= (2n+\alpha-2)\cdot \left((n-2)(1+2/\alpha) + 1\right).
\end{align*}

\noindent Then, for sufficiently large $n$, the ratio between he social costs of the NE network $S_n$ and the optimum $S_n^*$ is $\frac{\alpha+2}{2}-\varepsilon$. 
\end{proof}

\subsection{Points in $\mathbf{R}^d$}\label{subsec_points}
In this section we consider the \MNCG with the assumption that all nodes are points in $\mathbf{R}^d$ and that distances are measured via the $p$-norm, i.e., for any two points $u=(u_1,\ldots,u_d), v=(v_1,\ldots,v_d)$ the weight of the corresponding edge between them is defined as $$w_p(u,v):= \left(\sum\limits_{i=1}^d{|u_i-v_i|^p}\right)^{1/p}.$$ Further, we omit the subscript $p$ if its value does not play any role. 

\paragraph{Hardness}
We start with investigating the hardness of computing the best response of an agent in the \RDNCG.
 
\begin{theorem}\label{th:points_in_he_plane_BR_hardness}
It is NP-hard to compute a best response of an agent in the \RDNCG under any $p$-norm.
\end{theorem}
\begin{proof}
We perform the proof by a reduction from the Minimum Set Cover problem analogously to the proof of the Theorem~\ref{thm_tree_hardness}. 
We define the corresponding instance of the best response problem in the \RDNCG with $\alpha=1$ as follows: Consider a graph $G=(V,E)$ such that $$V = \{u\}\cup \{a_1,\ldots,a_m\}\cup\{b_1,\ldots,b_m\}\cup \{p_1,\ldots,p_k\}$$ and $$E=\bigcup\limits_{i=1}^m\{(b_i,u), (b_i,a_i)\}\cup \bigcup\limits_{i=1}^m\bigcup\limits_{p_j\in X_i}\{(a_i,p_j)\},$$ where each $p_j$ represents one element of the universe $U$ and each $a_i$ corresponds to one subset $X_i$. We locate nodes on the plane such that all points $a_1,\ldots,a_m$ are at the same distance $L$ from $u$ and equally spaced on the circle segment of length equal to some arbitrary small value $\varepsilon >0$. All points $p_1,\ldots,p_k$ are equispaced on the circle segment of the same length $\varepsilon$ and are at distance $2L$ from $u$.  We assume throughout the proof that $L >>\varepsilon$ holds. In addition, we have a set of nodes $\{b_1,\ldots,b_m\}$ at distance $\frac{1}{2}(L-\beta)$, where $\frac{1}{3}L > \beta > k\varepsilon$, such that each $b_i$, for $1\leq i \leq m$ lies on the line through the nodes $u$ and $a_i$ and is connected to nodes $u$ and $a_i$. By construction we get that $d_G(b_i,a_i) = \frac{1}{2}(L-\beta)+L$, $d_G(u,a_i) = 2L-\beta$ and $d_G(u,p_j) = 3L- \beta$. Moreover, we assume that each edge $(b_i,u)$ is owned by the respective node $b_i$. Finally, note that agent $u$ does not own any edges in $G$. 
See Figure~\ref{fig:NPhard} for the illustration of the constructed graph.
\begin{figure}[h!]
\center{\includegraphics[width=0.6\textwidth]{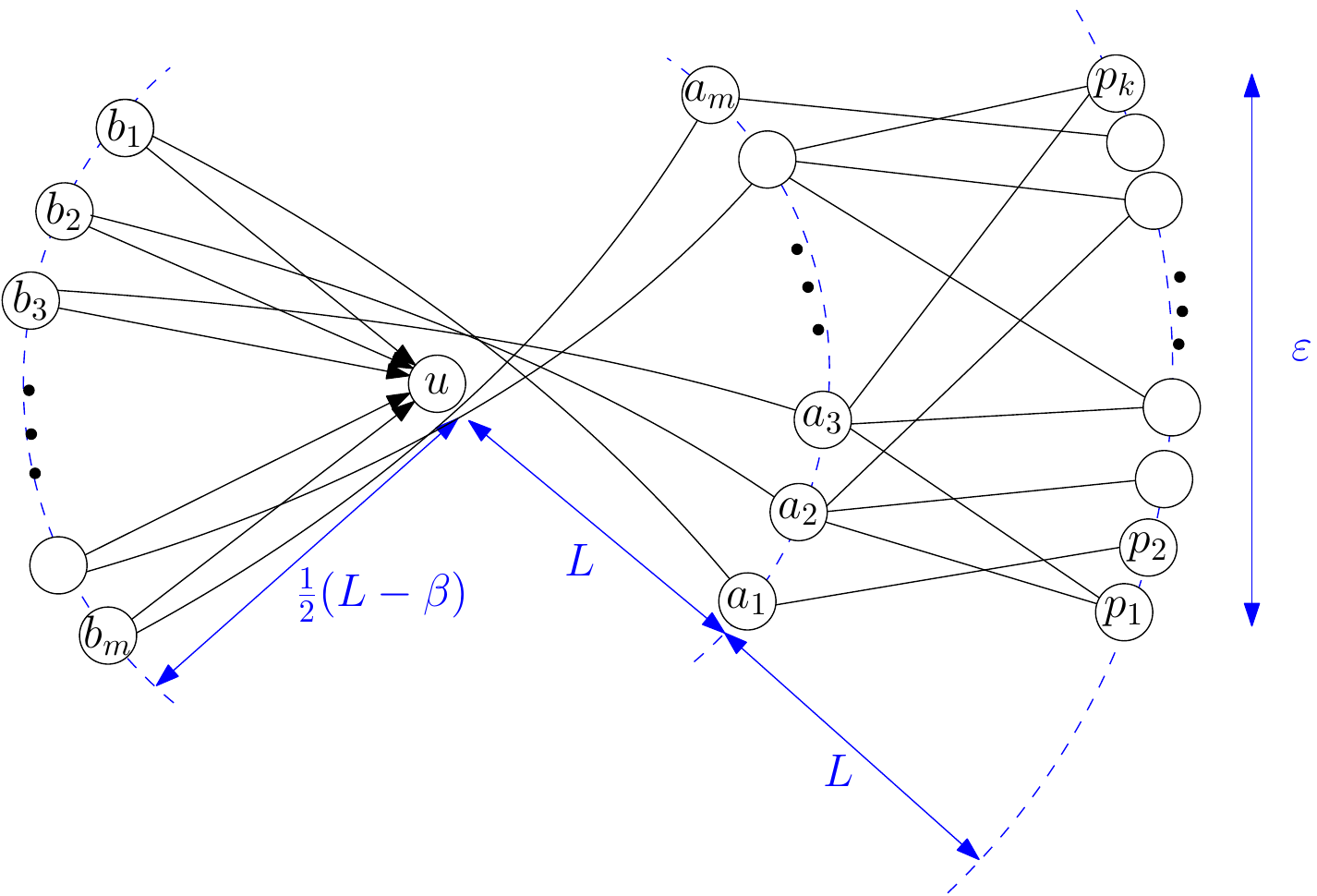}}
\caption{Illustration of the construction used in the reduction.}\label{fig:NPhard}
\end{figure} 

\noindent We claim that the best response of the agent $u$ corresponds to the solution of the Minimum Set Cover problem. 

 Consider a best response strategy $S_u^*$ of agent $u$ and the corresponding network $G^*$ which is the network $G$ augmented by the edges which agent $u$ buys according to strategy $S_u^*$. Thus, $G^* = (V,E^*)$, where $E^* = E(G) \cup \bigcup_{v\in S_u^*}\{(u,v)\}$. 
 
 First of all, we show that $S_u^*\neq\emptyset$ must hold, since adding at least one edge $(u,a_i)$ to network $G$ is already an improvement for agent $u$. This edge costs $L$ and decreases agent $u$'s distance to $a_i$ by $L-\beta$. Moreover, since $X_i$ is non-empty there must be a neighboring $p$-node of node $a_i$. Let $p_j$ be this node. Buying edge $(u,a_i)$ changes agent $u$'s distance to $p_j$ from at least $3L-\beta$ to at most $2L+\varepsilon$, which is an improvement of at least $L-\beta-\varepsilon$. Hence, the total improvement for agent $u$ is at least $L-2\beta-\varepsilon > 0$.
 
 Now we prove that agent $u$ always prefers buying edges to $a_i$ nodes over edges to $p_j$ nodes. Assume that $p_j\in S_u^*$ and let $a_r$ be a set node which is adjacent to $p_j$. Since $\bigcup X_i = U$, such a node $a_r$ must exist by construction.
If there is a set node $a_i\in S_u^*$ which is adjacent to $p_j$ in $G^*$, then agent $u$ can simply delete the edge $(u,p_j)$ because it decreases her edge cost by at least $2L$ and increases her distance cost by at most $\varepsilon$, since only the distance to the nodes $p_j$ can increase by at most $\varepsilon$, since for any $1\leq r\leq m$ holds $w(a_i,a_r) \leq \varepsilon$.

Since $L>>\varepsilon$, deleting $(u,p_j)$ would be an improving move for agent $u$. 
If there is no node $a_i\in S_u^*$ such that $(a_i,p_j)\in E(G^*)$, then the swap of the edge $(u,p_j)$ to $(u,a_r)$ improves agent $u$'s cost by at least $2L-\varepsilon-\beta$ since this move improves $u$'s distance to at least node $a_r$ by $L-\beta$, increases $u$'s distance only to the nodes $p_j$ by $\varepsilon$, respectively, and improves $u$'s total edge cost by~$L$.
Hence, $S_u^*$ cannot contain $p_j$ nodes.

Next, we show that every $p_i$ node is adjacent to some node $a_i \in S_u^*$, i.e., that the corresponding set of subsets $\{X_i \mid a_i \in S_u^*\}$ is a set cover of $U$. For the sake of contradiction, assume that there is a node $p_j$ for which there is no node $a_i \in S_u^*$ such that $(a_i,p_j) \in E(G^*)$. Clearly, there must be a path from $u$ to $p_j$ in $G^*$ since otherwise agent $u$ would have infinite cost. Let $a_r \notin S_u^*$ be any set node, for which $(a_r,p_j) \in E(G^*)$. Such a node $a_r$ must exist, since $\bigcup X_i = U$. Thus, we have that $d_{G^*}(u,p_j) \geq 3L-\beta$, since there is a path from $u$ to $p_j$ via $b_r$ and $a_r$. 
We claim that agent $u$ could buy the edge $(u,a_r)$ and thereby strictly decrease her cost. The edge $(u,a_r)$ costs $L$ and decreases agent $u$'s distances to $a_r$ by $L-\beta$ and to each of the nodes $p_j$ by at least $L-\beta$. Thus, this yields a cost decrease for agent $u$. Note, that this implies that agent $u$ can improve on any strategy $S_u$, where the corresponding set of subsets $\{X_i \mid a_i \in S_u\}$ does not cover all elements of $U$.     
Thus, the set of subsets $\{X_i \mid a_i \in S_u^*\}$ must be a set cover of $U$. 

We finish the proof by showing that the best response strategy of agent $u$ corresponds to a minimum set cover of the given set cover instance. For this, consider two arbitrary strategies $S_u^1$ and $S_u^2$ of agent $u$, such that the corresponding sets $\{X_i \mid a_i \in S_u^1\}$ and $\{X_i \mid a_i \in S_u^2\}$ both cover all elements of $U$. Now we show that if $|S_u^1| < |S_u^2|$ then $u$'s cost with strategy $S_u^1$ is strictly less than $u$'s cost with strategy $S_u^2$. This implies that agent $u$'s best response strategy $S_u^*$ corresponds to a minimum set cover.

Let $\Delta = |S_u^2| - |S_u^1|$. Hence, the difference between agent $u$'s edge cost with strategy $S_u^1$ and $u$'s edge cost with strategy $S_u^2$ is exactly $-\Delta\cdot L$. Since both strategies correspond to set covers and since $w(a_1,a_m) = \varepsilon$, the distances of $u$ to any $p_j$  node under the strategies $S_u^1$ and $S_u^2$ can differ by at most $\varepsilon$. Moreover, with strategy $S_u^1$ agent $u$ has distance $L$ to exactly $|S_u^1|$ many $a_i$ nodes and distance $2L-\beta$ to all the other $a_i$ nodes. Analogously, with strategy $S_u^2$ agent $u$ has distance $L$ to $|S_u^2|$ many $a_i$ nodes and distance $2L-\beta$ to the other $a_i$ nodes. 
Thus, the total difference in agent $u$'s cost with strategy $S_u^1$ compared with strategy $S_u^2$ for agent $u$ is $ -\Delta\cdot L + k\varepsilon + \Delta(L-\beta) = -\Delta\beta + k\varepsilon < 0$, where the inequality holds since $\Delta \geq 1$ and by construction we have $\beta > k\varepsilon$.   
\end{proof}

\paragraph{Dynamic Properties} Also for the \RDNCG we investigate whether best response dynamics are guaranteed to converge, i.e. if the game has the finite improvement property.
\begin{theorem}\label{th:points_in_the_plane_best_resp_cycle}
The \RDNCG with the 1-norm does not have the finite improvement property. 
\end{theorem}
\begin{proof}[Proof(sketch)]
We prove the statement by providing best response cycle, shown in Fig.~\ref{fig:best_resp_points_in_the_plane}.
\begin{figure}[h!]
\center{\includegraphics[width=0.8\textwidth]{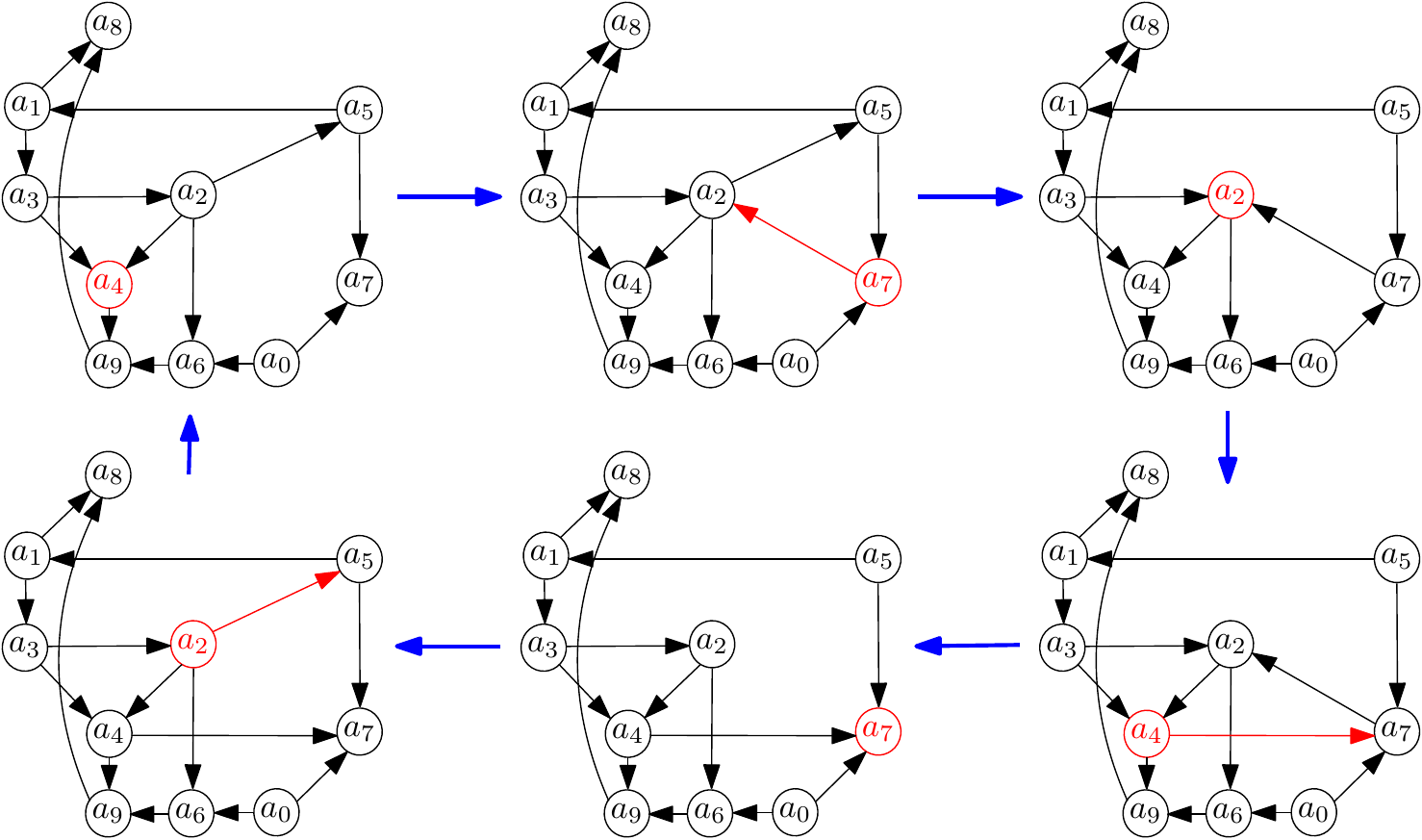}}
\caption{Best response cycle for the \RDNCG with 1-norm.}\label{fig:best_resp_points_in_the_plane}
\end{figure} 
We use the following node positions: $a_0 = (3, 0), a_1= (0,3), a_2= (2,2), a_3= (0,2), a_4= (1,1), a_5= (4,3), a_6= (2,0), a_7= (4,1), a_8= (1,4), a_9= (1,0)$. Distances are measured by the 1-norm.
\end{proof}
\noindent We are convinced that the above best response cycle can be adapted to arbitrary $p$-norms.
\begin{conjecture}
The \RDNCG with any p-norm does not have the finite improvemet property. 
\end{conjecture}

\paragraph{Price of Anarchy}
From Theorem~\ref{thm:metric_PoA} it follows that in the \RDNCG the \PoA is at most $\frac{\alpha+2}{2}$. It turns out that settling the PoA for the \RDNCG is a challenging problem. We prove some first steps in this direction and show that the PoA approaches the upper bound for the 1-norm if the number of dimensions grows. 

We start with a lower bound which is strictly larger than 1 for the \PoA in case of an arbitrary $p$-norm and independent of number of nodes $n$ and dimension $d$.

\begin{lemma} \label{lem:LB_PoA_weighted_path}
The PoA in the \RDNCG is strictly larger than 1.
\end{lemma}
\begin{proof}
To prove the claim we show a NE graph and the corresponding optimal graph in the 1-dimensional space, i.e., path metric. It is clear that the obtained result holds for arbitrary $d\geq 1$ and independent on the $p$-norm.   

Consider a path graph $P_{n+1}$ on $n+1$ nodes $\{v_0,v_1,\ldots,v_n\}$. We arrange lengths of the edges as follows: $w(v_0,v_1)=1$, and $\forall i\in\{2,\ldots,n\},\  w(v_{i-1},v_i)=\frac{2}{\alpha}\cdot\left(1+\frac{2}{\alpha}\right)^{i-2}$. Then the host graph $H$ is a complete graph containing $P_{n+1}$ such that for any pair of nodes $u, v$  $w(u,v)=d_{P_{n+1}}(u,v)$. Clearly, $P$ corresponds to a social optimum.

\begin{figure}[h!]
\center{\includegraphics[width=0.7\textwidth]{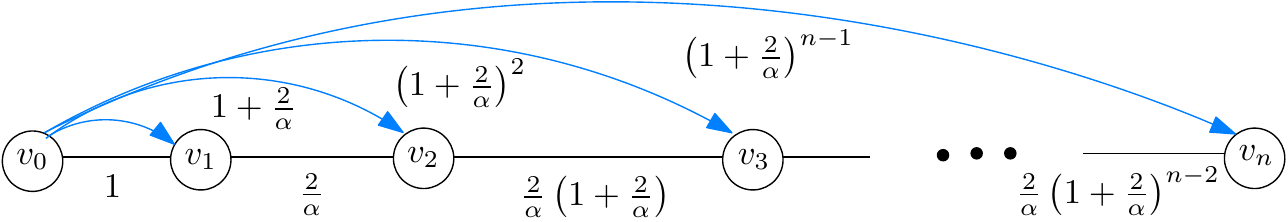}}
\caption{Lower bound graph. The black graph corresponds to the social optimum $P_{n+1}$, the blue graph $S_{n+1}$ is in NE.}\label{fig:LB_PoA_general_alpha_path}
\end{figure}

Consider a star graph $S:= S_{n+1}$ on the same set of nodes with the following edge weights $w(v_{0},v_{i})=d_{P}(v_{0},v_i)=\sum\limits_{j=1}^{i}{w(v_{j-1},v_j)} =1+\frac{2}{\alpha}\cdot \sum\limits_{j=2}^{i}{\left(1+\frac{2}{\alpha}\right)^{j-2}}=1+\frac{2}{\alpha}\cdot \left(\left(1+\frac{2}{\alpha}\right)^{i-1}-1\right)\big/\frac{2}{\alpha}= \left(1+\frac{2}{\alpha}\right)^{i-1}$. See figure~\ref{fig:LB_PoA_general_alpha_path} for the construction.

First, we show that the star graph $S_n$ is in NE. Since no deletions or swaps are possible, we need to prove that no addition of any new edge is profitable for its owner. Indeed, consider a player $v_i$, which tends to buy an edge $(v_i,v_j)$ such that $j\leq i-1$. This move decreases distance cost of the player by $w(v_i,v_0)+w(v_0,v_j)-w(v_i,v_j)=2\,w(v_0,v_j)=2\left(1+\frac{2}{\alpha}\right)^{j-1}$. At the same time, the edge cost increases by $\alpha\cdot w(v_i,v_j) = \alpha\cdot \sum\limits_{k=j+1}^i{w(v_{k-1},v_k)}=\alpha\cdot \frac{2}{\alpha}\left(1+\frac{2}{\alpha}\right)^{j-1}\cdot\left(\left(\frac{2}{\alpha}+1\right)^{i-j}-1\right)\big/\frac{2}{\alpha}=\alpha\cdot \left(1+\frac{2}{\alpha}\right)^{j-1}\left(\left(\frac{2}{\alpha}+1\right)^{i-j}-1\right)\geq \alpha\cdot \left(1+\frac{2}{\alpha}\right)^{j-1}\cdot\left(\frac{2}{\alpha}+1-1\right)\geq  \alpha\cdot \left(1+\frac{2}{\alpha}\right)^{j-1}\cdot \frac{2}{\alpha}= 2\cdot(1+\frac{2}{\alpha})^{j-1}$. Thus, this move is not an an improvement for the player $v_i$. If $j>i$, distance cost is $2\left(1+\frac{2}{\alpha}\right)^{i-1}$, whereas the edge cost is $\alpha\cdot w(v_i,v_j) = \alpha\cdot \sum\limits_{k=i+1}^{j}{w(v_{k-1},v_k)}=\alpha\cdot \left(1+\frac{2}{\alpha}\right)^{i-1}\left(\left(\frac{2}{\alpha}+1\right)^{j-i}-1\right)\geq \alpha\cdot \left(1+\frac{2}{\alpha}\right)^{i-1}\left(\frac{2}{\alpha}+1-1\right)\geq  \alpha\cdot \left(1+\frac{2}{\alpha}\right)^{i-1}\cdot \frac{2}{\alpha}= 2\cdot(1+\frac{2}{\alpha})^{i-1}$. Thus, the star graph is in NE. The social cost of $S_n$ is
$cost(S_{n+1}) = (2n+\alpha)\cdot \sum\limits_{(u,v)\in E(S_n)}{w(u,v)}=(2n+\alpha)\sum\limits_{i=1}^n{\left(1+\frac{2}{\alpha}\right)^{i-1}}=(2n+\alpha)\cdot\frac{\alpha}{2}\left(\left(1+\frac{2}{\alpha}\right)^{n}-1\right).$

The social optimum is a path graph $P_{n+1}$. An edge cost is $\alpha\cdot w(v_0,v_n)=\alpha\left(1+\frac{2}{\alpha}\right)^{n-1}$. To calculate the distance cost we count for each edge how many shortest paths it participates, i.e., its betweenness centrality.
\begin{align*}
d(P_n)&=2\sum\limits_{i=1}^n{w(v_{i-1},v_i)i(n-i+1)} =2\sum\limits_{k=0}^{\lfloor n/2\rfloor}{(n-2k)\sum\limits_{i=k+1}^{n-k}{w(v_{i-1},v_i)}}\\
&=2n\sum\limits_{i=1}^{n}{w(v_{i-1},v_i)} + 2\sum\limits_{k=1}^{\lfloor n/2\rfloor}{(n-2k)\sum\limits_{i=k+1}^{n-k}{w(v_{i-1},v_i)}}\\
&=2n\cdot \left(1+\frac{2}{\alpha}\right)^{n-1} + \frac{4}{\alpha}\sum\limits_{k=1}^{\lfloor n/2\rfloor}{(n-2k)\sum\limits_{i=k+1}^{n-k}{\left(1+\frac{2}{\alpha}\right)^{i-2}}}\\
&=2n\cdot \left(1+\frac{2}{\alpha}\right)^{n-1} +2\sum\limits_{k=1}^{\lfloor n/2\rfloor}{(n-2k)\left(1+\frac{2}{\alpha}\right)^{k-1}\left(\left(1+\frac{2}{\alpha}\right)^{n-2k}-1\right)}\\
&\leq 2n\cdot \left(1+\frac{2}{\alpha}\right)^{n-1} + 2n\sum\limits_{k=1}^{\lfloor n/2\rfloor}{\left(1+\frac{2}{\alpha}\right)^{n-k-1}}\\
&=2n\cdot \left(1+\frac{2}{\alpha}\right)^{n-1} + (\alpha + 2)\cdot n\cdot{\left(1+\frac{2}{\alpha}\right)^{n-2}\left(1-\left(1+\frac{2}{\alpha}\right)^{-\lfloor n/2\rfloor}\right)}\\
& < 2n\cdot \left(1+\frac{2}{\alpha}\right)^{n-1} + (\alpha + 2)\cdot n\cdot \left(1+\frac{2}{\alpha}\right)^{n-2}\cdot \left(1-\left(1+\frac{2}{\alpha}\right)^{-1} \right)\\
& \leq 2n\cdot \left(1+\frac{2}{\alpha}\right)^{n-1} + 2n\cdot \left(1+\frac{2}{\alpha}\right)^{n-2}.
\end{align*}

Thus, an upper bound of the social cost of the social optimum is $$cost(P_n) < (2n+\alpha)\cdot \left(1+\frac{2}{\alpha}\right)^{n-1}\cdot\left(1 + \left(1+ \frac{2}{\alpha}\right)^{-1} \right).$$ Therefore, for sufficiently large $n$ the \PoA is at least:
\begin{align*}
\frac{cost(S_n)}{cost(P_n)} &>\frac{\frac{\alpha}{2}(2n+\alpha)\left(\left(1+\frac{2}{\alpha}\right)^n-1\right)}{(2n+\alpha)\cdot \left(1+\frac{2}{\alpha}\right)^{n-1}\cdot\left(1 + \left(1+ \frac{2}{\alpha}\right)^{-1} \right)}\\
&= \frac{\frac{\alpha}{2}\left(1+\frac{2}{\alpha}-1\right)\sum\limits_{i=0}^{n-1}\left(1+\frac{2}{\alpha}\right)^i}{\left(1+\frac{2}{\alpha}\right)^{n-1}\cdot\left(1 + \left(1+ \frac{2}{\alpha}\right)^{-1} \right)}\\
&= \frac{\sum\limits_{i=0}^{n-1}\left(1+\frac{2}{\alpha}\right)^i}{\left(1+\frac{2}{\alpha}\right)^{n-1} + \left(1+\frac{2}{\alpha}\right)^{n-2}}\geq 1 \text{ for } n\geq 2.
\end{align*}

Hence, $PoA > 1$.
\end{proof}
\noindent For a small number of players the lower bound for the \PoA can be improved, which is shown in the next lemma.

\begin{theorem}\label{thm:LB_PoA_Rd_for_arb_p_norm}
In the \RDNCG under any $p$-norm with $p\geq 1$ the PoA is at least $$\cfrac{3\alpha^3 + 24\alpha^2 + 40\alpha + 24}{\alpha^3 + 10\alpha^2 + 32\alpha + 24}\,.$$
\end{theorem}
\begin{proof}
We prove the claim for the 2-dimensional case, then the result immediately follows for higher dimensions.

Consider the same construction as in the proof of the lemma~\ref{lem:LB_PoA_weighted_path} restricted on 4 nodes $v_0,v_1,v_2$, $v_3$. The ratio between the cost of the star graph, which is in NE, and the social optimum is $$\frac{cost(S_4)}{cost(P_4)}=\frac{(6+\alpha)\cdot \frac{\alpha}{2}\cdot\left(\left(1+\frac{2}{\alpha}\right)^3-1\right)}{(6+\alpha)\left(1+\frac{2}{\alpha}\right)^{2} + \frac{4}{\alpha}} = \frac{3\alpha^3 + 24\alpha^2 + 40\alpha + 24}{\alpha^3 + 10\alpha^2 + 32\alpha + 24}$$
\end{proof}
\noindent In contrast to other $p$-norms, where $p\geq 2$, the 1-norm allows us to embed a reduced version of our lower bound construction from the \TNCG. With increasing number of dimensions we can embed more and more of our construction. The following lemma shows that for arbitrary large $d$ the lower bound of the \PoA approaches the upper bound of $\frac{\alpha+2}{2}$.
    
\begin{theorem}\label{thm:LB_PoA_Rd_for_1_norm}
In a 1-norm $d$-dimensional space the PoA is at least $1+\cfrac{\alpha}{2+\alpha/(2d-1)}$. 
\end{theorem}
\begin{proof}
Consider a set of $n=2d+1$ points $v_0 = (0,\ldots,0), v_1 = (1,0,\ldots,0)$, $v_2 = (-\frac{2}{\alpha},0\ldots,0)$, $v_i = (\underbrace{0,\ldots,0}_{i-1},\frac{2}{\alpha},0,\ldots,0), v_{i+d-1} = (\underbrace{0,\ldots,0}_{i-1},-\frac{2}{\alpha},0,\ldots,0)$ for $i\in \{2,\ldots,d\}$. 

\begin{figure}[h]
\center{\includegraphics[width=0.7\textwidth]{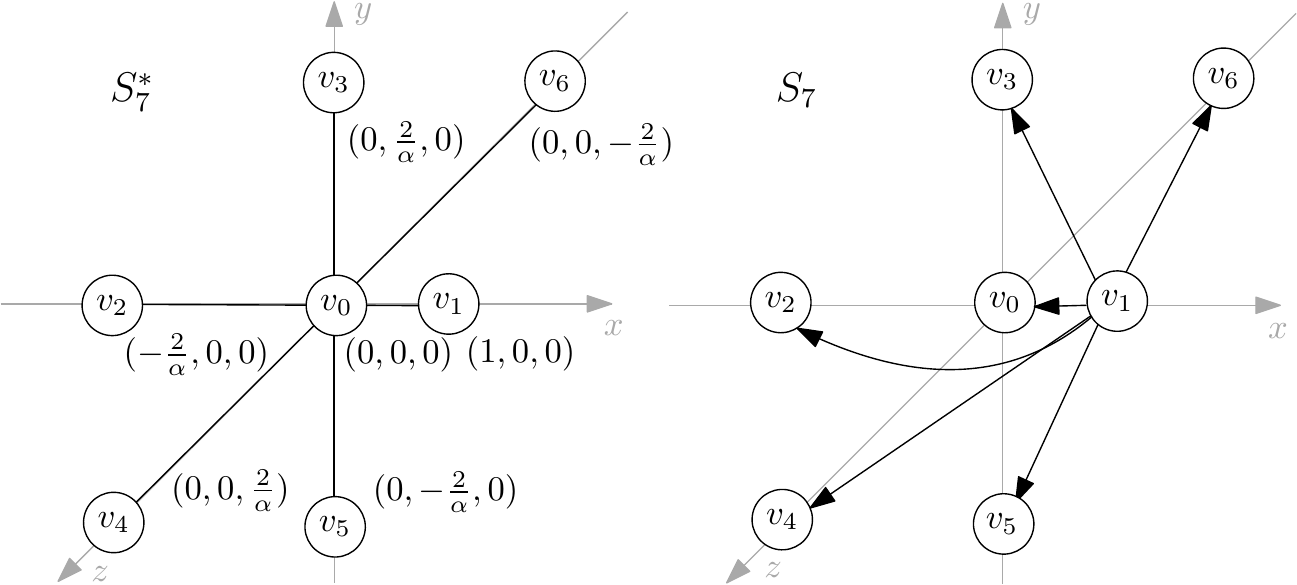}}
\caption{Construction of the optimal graph $S^*_7$ and the stable graph $S_7$ in 3-dimensional 1-norm space.}\label{fig:LB_PoA_3D}
\end{figure}

\noindent On the one hand it is easy to see that the star graph $S^*_n$ with a center in $v_0$ is an optimal network. On the other hand, the star graph $S_n$ with its center in $v_1$ such that $v_1$ is the owner of all edges in the star, is in NE. Indeed, since the distances are measured via the 1-norm, the construction is exactly the same as in the proof of Theorem~\ref{thm:LB_PoA_tree_metric}. See Fig.~\ref{fig:LB_PoA_3D} for the construction in $\mathbf{R}^3$. Thus, the \PoA is bounded by:
\begin{align*}
\frac{cost(S_n)}{cost(S_n^*)} &= \frac{\left(2\cdot(2d+1)+\alpha - 2\right)\left(\frac{\alpha+2}{\alpha}\cdot (2d-1) + 1\right)}{\left(2\cdot(2d+1)+\alpha - 2\right)\left(\frac{2}{\alpha}\cdot (2d-1) + 1\right)}= 1+\cfrac{\alpha}{2+\cfrac{\alpha}{2d-1}}.
\end{align*}
\end{proof}

\newpage
\section{General Weighted Host Graphs}

\begin{theorem}\label{thm_general_PoA}
The PoA for general instances is at most $\left(\frac{\alpha+2}{2}\right)^2$.
\end{theorem}
\begin{proof}
Let $G$ be a NE and let $u$ and $v$ be two distinct vertices. Let $x$ and $x^*$ be two Boolean variables such that $x=1$ if and only if $(u,v)$ is an edge of $G$ and $x^*=1$ if and only if $(u,v)$ is an edge of the social optimum $\OPT$. We prove the claim by showing that 
\[
\sigma:=\frac{\alpha \cdot w(u,v)x+2d_G(u,v)}{\alpha \cdot w(u,v)x^*+2d_{\OPT}(u,v)}\leq \left(\frac{\alpha+2}{2}\right)^2.
\]
The only case we have to prove is when $w(u,v) > d_H(u,v)$, $x=1$, and $x^*=0$  as the proof for all the other cases is identical to the proof of Theorem~\ref{thm:metric_PoA} for \MNCG, where  $\sigma \leq \frac{\alpha+2}{2}$ . 

We prove the claim by showing that $w(u,v) \leq \frac{\alpha+2}{2}\cdot d_H(u,v) \leq \frac{\alpha+2}{2}\cdot d_{\OPT}(u,v)$. The proof is by contradiction. W.l.o.g., we assume that $u$ is the owner of the edge $(u,v)$. Let $P$ be a shortest path between $u$ and $v$ in $H$. Since $w(u,v) > d_H(u,v)$, $P$ contains two or more edges. Furthermore, some edges of $P$ are missing from $G$, as otherwise player $u$ would remove the edge $(u,v)$ without increasing the distance towards $v$ and thus improving on her cost function. Let $k$ be the number of edges of $P$ that are not in $G$. We consider the edges of $P$ in order from $u$ to $v$ and we denote by $(x_i,y_i)$ the $i$-th edge of $P$ that is not in $G$. We observe that $x_1$ might be equal to $u$.

First, we prove that $d_G(u,y_i) \leq d_G(u,x_i)+\frac{\alpha+2}{2}w(x_i,y_i)$, for every $i=1,\dots,k$. We divide the proof into two cases, according to whether $x_i=u$ or not.

We consider the case in which $x_i \neq u$. We observe that this case definitely occurs if $i>1$. The proof is by contradiction. For the sake of contradiction, we assume that $d_G(u,y_i) > d_G(u,x_i)+\frac{\alpha+2}{2}w(x_i,y_i)$. By the triangle inequality, $d_G(u,y_i) \leq d_G(u,x_i)+d_G(x_i,y_i)$. As a consequence, $d_G(x_i,y_i) > \frac{\alpha+2}{2}w(x_i,y_i)$. Since $G$ is a NE, player $y_i$ does not have incentive in buying the edge towards $x_i$. Therefore, $\alpha w(x_i,y_i)+d_{G+(x_i,y_i)}(x_i,y_i)+d_{G+(x_i,y_i)}(y_i,u) \geq d_G(x_i,y_i)+d_G(y_i,u)$, i.e., 
$\alpha w(x_i,y_i)+2w(x_i,y_i)+d_{G}(x_i,u) \geq d_G(x_i,y_i)+d_G(y_i,u)>\frac{\alpha+2}{2}w(x_i,y_i)+d_G(u,x_i)+\frac{\alpha+2}{2}w(x_i,y_i)$, i.e., $0 > 0$, a contradiction.

We consider the case in which $x_i=u$. We observe that this case might occur only if $i=1$. Since $G$ is a NE, player $u$ has no incentive in swapping the edge $(u,v)$ with the edge $(u,y_1)$. If we denote by $G'$ the graph obtained from $G$ and in which player $u$ swaps the edge $(u,v)$ with the edge $(u,y_1)$, we get $\alpha \cdot w(u,y_1)+d_{G'}(u,y_1)+d_{G'}(u,v) \geq \alpha w(u,v)+d_G(u,y_1)+d_G(u,v)$. By substituting $d_{G'}(u,y_1)=w(u,y_1)$, $d_{G'}(u,v) \leq w(u,y_1)+d_{G-(u,v)}(y_1,v)$, and $d_G(u,v)=w(u,v)$ and simplifying we have $(\alpha+1)w(u,v) \leq (\alpha+1)d_H(y_1,u)+d_{G-(u,v)}(y_1,v)$. If $d_{G-(u,v)}(y_1,v) = d_G(y_1,v)$, then, since $d_G(y_1,v)\leq d_G(y_1,u)+w(u,v)$, $\alpha\cdot w(u,v) \leq \alpha\cdot d_H(y_1,u) \leq \alpha\cdot d_H(u,v)$, i.e., $w(u,v)=d_H(u,v) \leq \frac{\alpha+2}{2}\cdot d_H(u,v)$, thus contradicting the fact that $w(u,v) > \frac{\alpha+2}{2}\cdot d_H(u,v)$. Therefore, we can assume $d_{G-(u,v)}(y_1,v) > d_G(y_1,v)$. As a consequence, $d_G(y_1,v)=d_G(y_1,u)+w(u,v)$. Since $G$ is a NE, player $y_1$ has no incentive in buying the edge towards $u$. If we denote by $G''$ the graph obtained from $G$ and in which player $y_1$ buys the edge $(u,y_1)$, then $\alpha \cdot w(u,y_1)+d_{G''}(y_1,u)+d_{G''}(y_1,v) \geq d_G(y_1,u)+d_G(y_1,v)$, i.e., $\alpha \cdot  w(u,y_1)+ w(u,y_1)+w(u,y_1)+w(u,v) \geq d_G(y_1,u)+d_G(y_1,u)+w(u,v)$, i.e., $d_G(y_1,u) \leq \frac{\alpha+2}{2}w(y_1,u)$.

Now, using induction on $i$, we prove that $d_G(u,y_i) \leq \frac{\alpha+2}{2}\cdot d_P(u,y_i)$. For the base case, clearly $d_G(u,x_1) = d_P(u,x_1)$ since no edge of $P$ between $u$ and $x_i$ is missing from $G$. Therefore, using the triangle inequality, we obtain $d_G(u,y_1) \leq d_G(u,x_1)+d_G(x_1,y_1) \leq d_P(u,x_1)+\frac{\alpha+2}{2}w(x_1,y_1) \leq \frac{\alpha+2}{2}d_P(u,y_1)$. For the inductive case, by induction, we have that $d_G(u,y_{i-1}) \leq \frac{\alpha+2}{2}d_P(u,y_{i-1})$. Since no edge of $P$ between $y_{i-1}$ and $x_i$ is missing from $G$, we have that $d_G(u,x_i) \leq d_G(u,y_{i-1})+d_G(y_{i-1},x_i) \leq \frac{\alpha+2}{2}d_P(u,y_{i-1}) + d_P(y_{i-1},x_i) \leq \frac{\alpha+2}{2} d_P(u,x_i)$. Therefore, using the triangle inequality, we have that $d_G(u,y_i) \leq d_G(u,x_i)+d_G(x_i,y_i) \leq d_P(u,x_i)+\frac{\alpha+2}{2}w(x_i,y_i) \leq \frac{\alpha+2}{2}d_P(u,y_i)$.

To show that $d_G(u,v) \leq \frac{\alpha+2}{2}d_H(u,v)$, we observe that $w(u,v) \leq d_G(u,v)$, as otherwise player $u$ would remove the edge towards $v$ thus improving on her cost function. Therefore, by the triangle inequality and since no edge of $P$ between $y_k$ and $v$ is missing from $G$, we obtain $w(u,v) \leq d_G(u,v) \leq d_G(u,y_k)+d_P(y_k,v) \leq \frac{\alpha+2}{2}d_P(u,y_k) + d_P(y_k,v) \leq \frac{\alpha+2}{2}d_H(u,v)$. 
\end{proof}

The proof technique used in Theorem~\ref{thm_general_PoA} cannot lead to a better bound on the PoA. Consider, as a host graph, a cycle of three edges of weight 0, 1, and $(\alpha+2)/2$. The social optimum is given by the path containing the two edges of weight 0 and 1, respectively; a NE is given by the path containing the edges of weight 0 and $(\alpha+2)/2$, respectively. Let $u$ and $v$ be the two endvertices of the edge of weight $(\alpha+2)/2$. The value $\sigma$, so as defined in the proof of Theorem~\ref{thm_general_PoA}, is exactly equal to $\big(\frac{\alpha+2}{2}\big)^2$. However, if we compute the ratio between the cost of the defined NE and the cost of the social optimum, we obtain the value $(\alpha+2)/2$ that coincides with the PoA for \MNCG (see Theorem~\ref{thm:metric_PoA}). Therefore, we close this section with the following very interesting conjecture, stating that the PoA for \GNCG should be the same as the PoA for \MNCG.
\begin{conjecture}\label{conj_PoA}
The PoA for the \GNCG is $\frac{\alpha+2}{2}$.
\end{conjecture}

\section{Conclusion}
In this paper we have analyzed the well-known Network Creation Game on weighted complete host graphs. We think this is a significant step towards a more realistic game-theoretic model for the decentralized creation of networks, like fiber-optic or overlay networks. We showed that the weighted version of these games behaves similarly to the unit-weight NCG in terms of the hardness of computing a best response and in its dynamic properties. However, the Price of Anarchy is radically different. Whereas in the original NCG the PoA is conjectured to be constant and actually proven to be constant for almost all $\alpha$, we have shown that the PoA, even for the restricted metric case of the \TNCG, is linear in $\alpha$. Since $\alpha$ is a parameter for adjusting the trade-off between edge cost and distance cost, this implies that for settings where the edge cost dominates, i.e. $\alpha$ is high, coordination is needed to guarantee socially efficient outcomes. 

For understanding the impact of coordination, the next step should be to analyze the Price of Stability, i.e., the social cost ratio of the best equilibrium network and the social optimum. Another challenging task is to prove or refute that pure Nash equilibria always exist and to find a way to guide the agents to stable states with preferably low social cost. Besides this, naturally our conjectures, most prominently Conjecture~\ref{conj_PoA}, call for further investigation.   

\section{Acknowledgment}
We thank our anonymous reviewers for their valuable suggestions. This work has been partly supported by COST Action CA16228 European Network for Game Theory (GAMENET).


\bibliographystyle{abbrv}
\bibliography{geometric_NCG}

\end{document}